\documentclass[11pt]{article}

\usepackage{amsmath}
\usepackage{amsthm}
\usepackage{amssymb}
\usepackage{algorithm}
\usepackage{subfigure}
\usepackage[dvipsnames]{xcolor}
\usepackage{graphicx}
\usepackage{subfig}
\usepackage{algpseudocode}
\usepackage{comment}
\usepackage{mathtools}
\usepackage{enumitem}
\usepackage{tikz}
\usepackage{url}
\usepackage{float}
\usepackage[margin=1in]{geometry}
\usepackage{qcircuit}
\usepackage{bbm}
\usepackage[noblocks]{authblk}
\usepackage{cleveref}
\usepackage{caption}
\usepackage{footmisc}

\newtheorem{theorem}{Theorem}[section]
\newtheorem{lemma}[theorem]{Lemma}
\newtheorem{definition}[theorem]{Definition}

\newtheorem{proposition}[theorem]{Proposition}
\newtheorem{corollary}[theorem]{Corollary}

\newcommand{\lrb}[1]{\left ( #1 \right )}
\newcommand{\lrsb}[1]{\left [ #1 \right ]}
\newcommand{\lrcb}[1]{\left \{ #1 \right \}}

\newcommand{\abs}[1]{\left | #1 \right |}
\newcommand{\norm}[1]{\left \| #1 \right \|}

\newcommand{\mysin}[1]{\operatorname{sin}\lrb{#1}}
\newcommand{\mycos}[1]{\operatorname{cos}\lrb{#1}}

\newcommand{\ket}[1]{\left| #1 \right\rangle}
\newcommand{\bra}[1]{\left\langle #1 \right|}
\newcommand{\ketbra}[2]{|#1\rangle \langle #2|}
\newcommand{\braket}[2]{\langle #1|#2\rangle }

\newcommand{\zo}{\{0,1\}}

\newcommand{\myexp}[1]{\operatorname{exp} \lrb{#1}}

\newcommand{\mylog}[1]{\operatorname{log} \lrb{#1}}
\newcommand{\logab}[2]{\operatorname{log}_{#1} \lrb{#2}}

\newcommand{\myln}[1]{\operatorname{ln} \lrb{#1}}

\newcommand{\myArg}[1]{\operatorname{Arg} \lrb{#1}}

\newcommand{\R}{\mathbb{R}}
\newcommand{\C}{\mathbb{C}}
\newcommand{\Z}{\mathbb{Z}}
\newcommand{\N}{\mathbb{N}}

\renewcommand{\P}[1]{\mathbb{P}\lrsb{#1}}
\newcommand{\E}[1]{\mathbb{E}\lrsb{#1}}

\newcommand{\defeq}{\coloneqq}

\newcommand{\mymax}[1]{\operatorname{max}\lrb{#1}}
\newcommand{\mymin}[1]{\operatorname{min}\lrb{#1}}

\DeclareMathOperator{\realpart}{Re}
\DeclareMathOperator{\imaginarypart}{Im}
\renewcommand{\i}{\mathbf{i}}

\newcommand{\myO}[1]{\mathcal{O}\lrb{#1}}
\newcommand{\mytO}[1]{\tilde{\mathcal{O}}\lrb{#1}}
\newcommand{\myOmega}[1]{\Omega\lrb{#1}}
\newcommand{\myTheta}[1]{\Theta\lrb{#1}}

\DeclareMathOperator{\erfc}{erfc}

\newcommand{\hl}[1]{#1}

\title{Efficient ground-state energy estimation and certification on early fault-tolerant quantum computers}

\author[1]{Guoming Wang}
\author[2]{Daniel Stilck Fran\c{c}a}
\author[3]{Gumaro Rendon}
\author[3]{Peter D. Johnson}

\affil[1]{Zapata Computing Canada Inc., Toronto, ON M5C 3A1, Canada}
\affil[2]{Univ Lyon, ENS Lyon, UCBL, CNRS, Inria, LIP, F-69342, Lyon Cedex 07, France}
\affil[3]{Zapata Computing Inc., Boston, MA 02110 USA}

\date{}

\begin{document}

\maketitle

\begin{abstract}
A major thrust in quantum algorithm development over the past decade has been the search for the quantum algorithms that will deliver practical quantum advantage first. Today’s quantum computers—and even early fault-tolerant quantum computers—are limited in the number of operations they can implement per circuit. We introduce quantum algorithms for ground-state energy estimation (GSEE) that accommodate this design constraint. 

The first algorithm estimates ground-state energies, offering a quadratic improvement on the ground state overlap parameter compared to other methods in this regime. The second algorithm certifies that the estimated ground-state energy is within a specified error tolerance of the true ground-state energy, addressing the issue of gap estimation that beleaguers several ground state preparation and energy estimation algorithms. We note, however, that the scaling of this certification technique is currently less favorable than that of the GSEE algorithm.

To develop the certification algorithm, we propose a novel use of quantum computers to facilitate rejection sampling. After a classical computer generates initial samples, the quantum computer is used to accept or reject these samples, resulting in a set of accepted samples that approximate draws from a target distribution. Although we apply this technique specifically for ground-state energy certification, it may find broader applications.

Our work pushes the boundaries of what operation-limited quantum computers can achieve, bringing the prospect of quantum advantage closer to realization.
\end{abstract}

\section{Introduction}

Estimating the ground-state energy of quantum many-body systems is a fundamental problem in condensed matter physics, quantum chemistry, materials science, and quantum information. This problem underpins basic processes in drug discovery and materials design \cite{aspuru2005simulated, cao2019quantum, elfving2020will, goings2022reliably}.

In ground-state energy estimation, the system of interest is represented by an \( n \)-qubit Hamiltonian \( H \) with an unknown spectral decomposition \( H = \sum_j E_j \ket{E_j} \bra{E_j} \), where \( E_0 < E_1 \le E_2 \le \dots \) are the eigenvalues of \( H \) and the \( \ket{E_j} \) are the orthonormal eigenstates of \( H \). The goal is to estimate the lowest eigenvalue \( E_0 \) of \( H \) within an additive error \( \epsilon \) with high probability. Given its importance, this problem has been extensively studied, and various methods have been proposed to address it. Among these, two main approaches to ground-state energy estimation (GSEE) have been explored.

\paragraph{GSEE assuming a fault-tolerance cost model}
The first approach to GSEE is based on the quantum phase estimation (QPE) algorithm \cite{kitaev1995quantum, nielsen2001quantum}, which often provides rigorous performance guarantees. In this approach, we assume the existence of an efficient procedure to prepare a state \( \ket{\psi} \) that has a non-trivial overlap with the ground state of \( H \) (i.e., \( |\braket{E_0}{\psi}| = \Omega(1/\operatorname{poly}(n)) \)), a reasonable assumption in many practical settings. For instance, Refs.~\cite{babbush2015chemical, tubman2018postponing, sugisaki2018quantum, mcardle2020quantum} present several methods for generating such states for quantum chemistry Hamiltonians. Given this ansatz state \( \ket{\psi} \), we can apply standard QPE on the time evolution of \( H \) with \( \ket{\psi} \) as the initial state. Assuming \( p_0 \coloneqq |\braket{E_0}{\psi}|^2 \), we can obtain an \( \epsilon \)-accurate estimate of \( E_0 \) using \( \mathcal{O}(p_0^{-1}) \) QPE circuits, each utilizing \( \mathcal{O}(\log(\epsilon^{-1} p_0^{-1})) \) ancilla qubits and implementing controlled time evolution of \( H \) for up to \( \mathcal{O}(\epsilon^{-1} p_0^{-1}) \) time. The overall runtime of the algorithm scales as \( \mathcal{O}(\epsilon^{-1} p_0^{-2}) \). Many variants of QPE have been proposed to reduce the total evolution time to \( \tilde{\mathcal{O}}(\epsilon^{-1} p_0^{-1}) \) \cite{knill2007optimal, nagaj2009fast, poulin2009sampling, dong2022ground}, minimize the maximum evolution time to \( \mathcal{O}(\epsilon^{-1} \log(p_0^{-1})) \) \cite{knill2007optimal, nagaj2009fast, poulin2009sampling, lin2022heisenberg, wan2022randomized, dong2022ground}, or limit the number of ancilla qubits to \( \mathcal{O}(1) \) \cite{berry2009perform, higgins2007entanglement, lin2022heisenberg, dong2022ground}. 

Despite these improvements, all these algorithms require implementing controlled time evolution of \( H \) for up to \( \mathcal{O}(\epsilon^{-1}) \) time, resulting in high gate counts when small \( \epsilon \) values are needed.
For example, Ref.~\cite{kim2022fault} estimated that the number of T gates required to implement ground-state energy estimation to sufficient accuracy for a class of small molecules is greater than $10^{10}$; the error rate of each gate must be significantly less than $10^{-10}$.
Today's quantum hardware and that of the near future are far more limited in the number of gates that can be implemented per circuit.
For instance, IBM’s \( 100 \times 100 \) challenge \cite{gambetta2022hundred} aims to achieve accurate expectation value calculations using circuits with approximately \( 10^4 \) gates, each having an error rate of \( 10^{-3} \). There is broad consensus that scaling from \( 10^4 \) to over \( 10^{10} \) operations requires large-scale quantum architectures capable of low-error fault-tolerant operations.

The GSEE algorithms above were developed with a fault-tolerant quantum computer in mind. It is by now standard to assume a \emph{cost model} where the overhead of reducing error rates (through fault-tolerant protocols) does not change the polynomial scaling of a quantum algorithm; that is, the cost is logarithmic in the number of operations.
Under this model, it is generally preferable to reduce algorithm runtime at the cost of an increased number of operations per circuit.
For this reason, methods for ground-state energy estimation have been developed \cite{lin2020near} that increase the number of ancilla qubits and operations to reduce the overall runtime.

\paragraph{GSEE assuming a pre-fault-tolerance cost model}
The second approach to GSEE is the variational quantum eigensolver (VQE) \cite{peruzzo2014variational, aspuru2005simulated, mcclean2016theory}, initially conceived for a cost model without fault-tolerant protocols, where gate error rates are fixed. For a given circuit error rate, this limits the maximum number of quantum operations per circuit. VQE uses an ansatz circuit with a limited number of operations to approximate the ground state of \( H \). Although suitable for near-term implementation, VQE lacks performance guarantees for two reasons. First, the accuracy of its output depends on the representational power of the chosen variational ansatz, which is often difficult to assess. Second, finding the optimal ansatz parameters requires solving a non-convex optimization problem, which can be challenging~\cite{Bittel2021}. Furthermore, since VQE repeatedly evaluates the energy of the ansatz state (either via direct sampling or more sophisticated methods \cite{wang2021minimizing, koh2022foundations}), the algorithm becomes impractical for industrially-relevant chemical system sizes, as shown in recent work~\cite{gonthier2020measurements}, and is more sensitive to noise than previously anticipated~\cite{StilckFrana2021,DePalma2023,quek2024exponentially,wang2021noise}.

\paragraph{GSEE assuming a limited-number-of-operations cost model}
The two models of ground-state energy estimation algorithms introduced so far were developed under two very different cost models: logarithmic-overhead error reduction via fault tolerance, and fixed error rates. These models correspond to future fault-tolerant architectures and today's near-term intermediate-scale devices, respectively. However, these cost models may not fully capture the capabilities and limitations of early fault-tolerant quantum computers \cite{fellous2021limitations}, motivating the development of models that interpolate between these two regimes. While rigorously developing such intermediate cost models is beyond the scope of our work, we assume a proxy for this model: \emph{any early fault-tolerant quantum computer will be limited in the number of logical operations it can implement per circuit while maintaining the total error rate below a specified threshold.} This assumption motivates the development of quantum algorithms that reduce the number of operations below the device's limit, albeit at the cost of increased runtime.

Several recent works have been developed under similar premises. The first was the so-called \( \alpha \)-VQE method \cite{wang2019accelerated}, which introduced a variable-depth amplitude estimation. Later, Wang et al. presented a variable-depth amplitude estimation algorithm that is robust to significant levels of error \cite{wang2021minimizing}. Similar approaches have been explored in the context of quantum algorithms for finance \cite{alcazar2022quantum, giurgica2022lowa}, with some of these methods implemented on quantum hardware \cite{katabarwa2021reducing, giurgica2022lowb}. Other algorithms targeting early fault-tolerant quantum computers aim to reduce the number of ancilla qubits compared to their traditional counterparts \cite{suzuki2020amplitude, lin2022heisenberg, zhang2022computing}. Additionally, low-depth quantum phase estimation algorithms have been proposed to address cases where the overlap between the initial state and the target eigenstate is large \cite{ding2022even, ni2023low}.

Recently, Wang et al.~\cite{wang2023quantum} applied this methodology to the task of ground-state energy estimation. They asked whether a method exists for estimating ground-state energy that requires fewer operations per circuit than previous approaches (i.e., the advantage of VQE) while providing rigorous performance guarantees (i.e., the advantage of traditional GSEE methods). Inspired by the recent work of Lin and Tong \cite{lin2022heisenberg}, they showed that, given a lower bound \( \Delta \) on the spectral gap of \( H \) around \( E_0 \), one can estimate \( E_0 \) within additive error \( \epsilon \) with high probability by using \( \tilde{\mathcal{O}}(\Delta^2 \epsilon^{-2} p_0^{-2}) \) quantum circuits. Each circuit evolves \( H \) for up to \( \mathcal{O}(\Delta^{-1} \operatorname{polylog}(\Delta \epsilon^{-1} p_0^{-1})) \) time and requires only a single ancilla qubit. Notably, the maximum evolution time is only poly-logarithmic in \( 1/\epsilon \), which is exponentially shorter than in previous methods. Furthermore, the parameter \( \Delta \) can be chosen\footnote{\label{ft:gap_comment}\hl{As described in previous work \cite{wang2023quantum}, it is possible to choose a value of $\Delta$ between $\Delta_{\rm true}$ and $\epsilon$. In practice, the exact value of $\Delta_{\rm true}$ is generally unknown, so a sufficiently accurate estimate of $\Delta_{\rm true}$ must be obtained.
In Ref.~\cite{wang2023quantum}, EOM-CCSD calculations provided gap estimates with accuracy within \( \sim10\% \) of the true value, which sufficed to set a lower bound on \( \Delta_{\rm true} \) and thereby reduce the circuit depth.}} anywhere between \( \Theta(\epsilon) \) and \( \Theta(\Delta_{\rm true}) \), where \( \Delta_{\rm true} \coloneqq E_1 - E_0 \). By tuning \( \Delta \), they achieved a class of GSEE algorithms that smoothly interpolate between an algorithm with circuit depth \( \mathcal{O}(\Delta_{\rm true}^{-1} \operatorname{polylog}(\Delta_{\rm true} \epsilon^{-1} p_0^{-1})) \) and runtime \( \tilde{\mathcal{O}}(\Delta_{\rm true} \epsilon^{-2} p_0^{-2}) \) and an algorithm with circuit depth \( \mathcal{O}(\epsilon^{-1} \operatorname{polylog}(p_0^{-1})) \) and runtime \( \tilde{\mathcal{O}}(\epsilon^{-1} p_0^{-2}) \). These algorithms, therefore, transition smoothly from having low circuit depth to achieving Heisenberg-limit scaling in runtime with respect to \( \epsilon \). Later, Ding and Lin \cite{ding2022even} recovered this result using a different approach.

Although Refs.~\cite{wang2023quantum} and \cite{ding2022even} significantly reduce the circuit depth for GSEE, one aspect of their algorithms remains  unsatisfactory: their runtimes are quadratic in \( 1/p_0 \). This implies that if the overlap between the input state \( \ket{\psi} \) and the ground state \( \ket{E_0} \) is small, these algorithms may take a long time to complete. A natural question arises: can the dependence on \( 1/p_0 \) be improved to linear while retaining the desirable features of these algorithms? 

In this work, we answer this question affirmatively. Specifically, we show that, given a lower bound \( \Delta \) on the spectral gap of \( H \) around \( E_0 \), one can estimate \( E_0 \) within additive error \( \epsilon \) with high probability by using \( \tilde{\mathcal{O}}(\Delta^2 \epsilon^{-2} p_0^{-1}) \) quantum circuits. Each circuit has a depth of \( \mathcal{O}(\Delta^{-1} \operatorname{polylog}(\Delta \epsilon^{-1} p_0^{-1})) \) and requires \( \mathcal{O}(1) \) ancilla qubits. These circuits involve only controlled time evolution of $H$ along with elementary quantum gates. Additionally, as in Ref.~\cite{wang2023quantum}, we can tune \( \Delta \) anywhere between \( \Theta(\epsilon) \) and \( \Theta(\Delta_{\rm true}) \). Doing so yields a class of GSEE algorithms that smoothly transition from circuit depth \( \mathcal{O}(\Delta_{\rm true}^{-1} \operatorname{polylog}(\Delta_{\rm true} \epsilon^{-1} p_0^{-1})) \) and runtime \( \tilde{\mathcal{O}}(\Delta_{\rm true} \epsilon^{-2} p_0^{-1}) \) to circuit depth \( \mathcal{O}(\epsilon^{-1} \operatorname{polylog}(p_0^{-1})) \) and runtime \( \tilde{\mathcal{O}}(\epsilon^{-1} p_0^{-1}) \). Thus, our algorithm can trade circuit depth for runtime, like the one in Ref.~\cite{wang2023quantum}, but with a shorter runtime.

\paragraph{ground-state energy certification via rejection sampling}
Although good estimates of the spectral gap are available for many physically relevant models~\cite{lee2023evaluating}, this is not always the case. Moreover, our algorithm can yield accurate ground-state energy estimates for evolution times \( t \ll \Delta^{-1} \) if the initial state has certain structural properties, such as large overlap with the ground state and minimal overlap with other low-energy states. Therefore, it is highly desirable to ensure that our algorithm provides reliable estimates even without knowledge of \( \Delta \) or when we have reason to believe the initial state is well-structured.

To this end, we devise a test that does not rely on knowledge of the spectral gap and only accepts a ground-state energy estimate if it is within a specified error tolerance of the true ground-state energy\footnote{Our test certifies that a ground-state energy estimate is \( \epsilon \)-close to the correct value, rather than certifying that the Hamiltonian's spectral gap is at least some \( \Delta_0 \). Furthermore, the test may reject correct estimates at depths smaller than the inverse of the spectral gap, but it only accepts correct estimates.}. Additionally, the test will accept with high probability if the evolution time is proportional to the inverse of the spectral gap. Our algorithm, combined with this certification test, is the first to deliver certifiably accurate energy estimates in the regime of depths \( o(\epsilon^{-1}) \). However, the test’s sample complexity of \( \tilde{\mathcal{O}}(\epsilon^{-4} p_0^{-3}) \) may make implementation challenging in some contexts.

A core component of our certification algorithm is a novel procedure for performing rejection sampling on a quantum computer. We use it to draw random values from a neighborhood around \( E_0 \). Specifically, suppose we wish to sample from a target Gaussian-like distribution \( \nu(x) \) that peaks around \( E_0 \), which may be challenging to do directly. Instead, we first draw a random value \( x \) from an alternative distribution \( \mu(x) \) and construct a unitary operator \( U_x(H) \) such that \( |\bra{0^k} U_x(H) \ket{0^k} \ket{\psi}|^2 \propto \frac{\nu(x)}{\mu(x)} \) for some integer \( k \). We then accept \( x \) if a measurement on the first \( k \) qubits of \( U_x(H) \ket{0^k}\ket{\psi} \) yields the outcome \( 0^k \); otherwise, we discard \( x \). This process is analogous to the standard rejection sampling algorithm. The operation \( U_x(H) \) can be implemented with a low-depth circuit using the recent technique from Ref.~\cite{dong2022ground}. Furthermore, the circuit requires only a few ancilla qubits, and the rejection sampling procedure is efficient, meaning that not many samples are discarded. Once we obtain a sufficient number of samples from the desired distribution \( \nu(x) \), we can certify the ground-state energy estimate based on the sample statistics.

The remainder of this paper is organized as follows. Section \ref{sec:main_results} formally defines the problem and summarizes the main results. Section \ref{sec:gsee_alg} introduces our GSEE algorithm, beginning with a description of the basic strategy and a proof of its correctness, followed by a detailed implementation and an analysis of its costs. In Section \ref{sec:numerical_results}, we numerically evaluate the algorithm's performance and compare it with an alternative approach of similar circuit depth and overall runtime. Section \ref{sec:certification_result} presents our method for performing rejection sampling on a quantum computer and demonstrates its application in devising an algorithm for certifying ground-state energy estimates. Finally, Section \ref{sec:conclusion} concludes the paper.

\section{Overview of the main results}
\label{sec:main_results}
In this section, we outline the problem to be solved and summarize the main results of this work.

\subsection{Problem formulation}
Let us begin with a formal definition of the ground-state energy Estimation (GSEE) problem:

\begin{definition}[GSEE]
Let $H$ be an $n$-qubit Hamiltonian with unknown spectral decomposition $H=\sum_{j=0}^{N-1} E_j \ket{E_j} \bra{E_j}$, where $N=2^n$, and the eigenvalues satisfy $0 \le E_0<E_1 \le E_2 \le ...\le E_{N-1} \le \pi$. Suppose we know a lower bound $\Delta>0$ on the spectral gap $E_1-E_0$ of $H$. Additionally, suppose we can prepare an $n$-qubit state $\ket{\psi}=\sum_{j=0}^{N-1} \gamma_j \ket{E_j}$, with an overlap with the ground state $\ket{E_0}$ satisfying $|\gamma_0|^2 \ge \eta$, where $\eta \in (0, 1)$ is known. 

Our goal is to estimate $E_0$ within additive error $\epsilon$ with probability at least $1-\delta$, for given $\epsilon>0$ and $\delta \in (0, 1)$. To achieve this, we use quantum circuits that involve controlled-$e^{\i Ht_j}$ operations for various $t_j$ values, along with additional elementary quantum gates. Each of these circuits acts on the state  $\ket{0^k}\ket{\psi}$ for some $k \ge 1$.
\label{def:gsee}
\end{definition}

We assume the ability to implement the controlled time evolution of \( H \), i.e., the unitary operation controlled-\( e^{\i H t} \coloneqq \ketbra{0}{0} \otimes I + \ketbra{1}{1} \otimes e^{\i H t} \) for any \( t \in \mathbb{R} \). We define the circuit depth of our algorithm as the maximal evolution time of \( H \) and the overall runtime of our algorithm as the total evolution time of \( H \), since Hamiltonian simulation often represents the dominant cost of the algorithm.

Our algorithm requires probabilistic implementations of various (nonunitary) functions of \( H \). Equivalently, we need to implement certain unitary operations that serve as block-encodings of specific functions of \( H \). Recall that block-encodings of linear opreators are formally define as follows:
\begin{definition}[block-encoding]
Let \( n, m \in \mathbb{N} \) and \( \alpha, \epsilon \in \mathbb{R}^+ \) be arbitrary. We say that an \( (n + m) \)-qubit unitary operator \( U \) is an \( (\alpha, m, \epsilon) \)-block-encoding of an \( n \)-qubit linear operator \( A \) if 
\begin{align}
    \norm{A - \alpha \left( \bra{0^m} \otimes I \right) U \left( \ket{0^m} \otimes I \right)} \le \epsilon.
\end{align}
\end{definition}
In this work, our block-encodings of $f(H)$ require a single ancilla qubit (i.e. $m=1$) and have a scaling factor $\alpha=1$, provided that $\norm{f(H)}\le 1$.

We employ the technique of \emph{quantum eigenvalue transformation of unitary matrices with real polynomials} (QET-U) \cite{dong2022ground} to implement these block-encoding operations. QET-U circuits have a specific structure: they interleave controlled (forward and backward) time evolutions of \( H \) with carefully chosen \( X \)-rotations on the ancilla qubit. In our implementation, we modify these circuits by inserting a \( Z \)-rotation on the ancilla qubit after each controlled time evolution of $H$. This modification allows us to effectively implement the controlled-\( e^{\i (H - \xi I)t} \) operation, which is necessary for realizing certain non-even functions of \( H \). In addition, we might need to feed the output of a modified QET-U circuit into a Hadamard test circuit to extract useful information from its measurement outcome. Figure \ref{fig:quantum_circuits} illustrates the quantum circuits used in this work. They involve only simple single-qubit operations besides the controlled time evolutions of $H$.

\begin{figure}[H]
    \centering
    \begin{displaymath}
    \Qcircuit @C=0.8em @R=1.2em {
        & & & &\\			
        \lstick{\ket{0}}
        &\gate{\mathrm{H}}	 &\ctrl{1}	& \gate{\mathrm{W}}
        & \gate{\mathrm{H}}			&\meter\\
        \lstick{\ket{\phi}} 	 & \qw & \gate{e^{\i H \tau}} 		 
        &\qw &\qw &\qw
    }		
    \end{displaymath}
        (a) \label{fig:hadamard_test}
        \begin{displaymath}
    \Qcircuit @C=0.6em @R=1.2em {
        & & & & & & & & & & & & &\\			
        \lstick{\ket{0}}
        &\gate{e^{\i \varphi_0 X}}	 &\ctrl{1}	& \gate{e^{\i\omega Z}}
        & \gate{e^{\i \varphi_1 X}}	&\ctrl{1}	& \qw & \push{\hspace{-5pt}...} &\ctrl{1} & \gate{e^{\i\omega Z}}
        & \gate{e^{\i \varphi_1 X}}	&\ctrl{1}& \gate{e^{-\i\omega Z}}
        & \gate{e^{\i \varphi_0 X}}	&\meter\\
        \lstick{\ket{\psi}} 	 & \qw & \gate{e^{\i H t}} 		 
        &\qw &\qw & \gate{e^{-\i H t}} & \qw & \push{\hspace{-5pt}...}		& \gate{e^{\i H t}}
        &\qw &\qw   & \gate{e^{-\i H t}} & \qw & \qw & \qw
    }		
    \end{displaymath}
        (b)

\caption{Quantum circuits used in this work. (a) Hadamard test circuit, where \( W = I \) or \( S^\dagger \), with \( S \) being the phase gate. The input state might be generated by a modified QET-U circuit. (b) Modified QET-U circuit. In addition to the \( X \)-rotations on the ancilla qubit used in the standard QET-U circuit, we insert \( Z \)-rotations on the same qubit to effectively implement controlled time evolutions of \( H - \xi I \) for any given \( \xi \in \R\).}
\label{fig:quantum_circuits}
\end{figure}

\subsection{Main results}

Our main results are threefold: a state-of-the-art early fault-tolerant algorithm for GSEE, a new method for performing rejection sampling on a quantum device, and techniques for certifying ground-state energy estimates. We will now informally summarize these technical contributions.

First, we develop a new GSEE algorithm that achieves a better depth-runtime tradeoff than previous algorithms:

\begin{theorem}[GSEE algorithm, informal]
Let $H$ be an $n$-qubit Hamiltonian with unknown spectral decomposition $H=\sum_{j=0}^{N-1} E_j \ket{E_j} \bra{E_j}$, where $N=2^n$, and the eigenvalues satisfy $0 \le E_0<E_1 \le E_2 \le ...\le E_{N-1} \le \pi$, and $E_1-E_0 \ge \Delta$ for a known $\Delta >0$. Moreover, suppose we can prepare an $n$-qubit state $\ket{\psi}$ such that $|\langle\psi|E_0\rangle|^2 \ge \eta$ for a known $\eta \in (0, 1)$. 

Then, assuming the ability to implement controlled time evolution of $H$, there exists an algorithm that, for any given $\epsilon>0$ and $\delta\in (0,1)$, estimates $E_0$ within additive error $\epsilon$ with probability at least $1-\delta$ by making use of 
$$\tilde{\mathcal{O}}\lrb{\eta^{-1} (\Delta^{2}\epsilon^{-2} + \mylog{\Delta^{-1}})\mylog{\delta^{-1}}}$$
 quantum circuits, where each circuit has depth proportional to $$\tilde{\mathcal{O}}\lrb{\Delta^{-1} \operatorname{log}(\eta^{-1}\Delta \epsilon^{-1})}.$$
     \label{thm:main_informal}
\end{theorem}

For ease of presentation, the theorems provide algorithm costs in big-O notation. However, interested readers can find the constant and logarithmic pre-factors in Algorithms \ref{alg:basic_gsee}, \ref{alg:gsee} and \ref{alg:gsee_ref}. Note that, as is standard, the Hamiltonian — and thus the ground-state energy estimate — will likely need to be normalized and shifted so that its spectrum lies within the interval \( [0, \pi] \).

We emphasize that in Theorem \ref{thm:main_informal}, $\Delta$ is merely a lower bound on the true spectral gap $\Delta_{\rm true} \defeq E_1-E_0$ of $H$, not necessarily $\Delta_{\rm true}$ itself. \hl{Thus, in principle, the algorithm parameter $\Delta$ can be tuned
between $\Theta(\epsilon)$ and $\Theta(\Delta_{\rm true})$. Defining parameter $\beta\in[0,1]$ and setting $\Delta = \Delta_{\textrm{true}}^{1-\beta}\epsilon^{\beta}$ then leads to the existence of a class of GSEE algorithms}:

\begin{corollary}[GSEE algorithms, informal]
Let $H$ be an $n$-qubit Hamiltonian with unknown spectral decomposition $H=\sum_{j=0}^{N-1} E_j \ket{E_j} \bra{E_j}$, where $N=2^n$, and the eigenvalues satisfy $0 \le E_0<E_1 \le E_2 \le ...\le E_{N-1} \le \pi$, with a spectral gap $\Delta_{\rm true}=E_1-E_0$. Moreover, suppose we can prepare an $n$-qubit state $\ket{\psi}$ such that $|\langle\psi|E_0\rangle|^2 \ge \eta$ for a known $\eta \in (0, 1)$. 

Then, assuming the ability to implement controlled time evolution of $H$, for every $\beta \in [0, 1]$, there exists an algorithm that, for any given $\epsilon>0$ and $\delta\in (0,1)$, estimates $E_0$ within additive error $\epsilon$ with probability at least $1-\delta$ by making use of 
$$\tilde{\mathcal{O}}\lrb{\eta^{-1} \lrb{\epsilon^{-2+2\beta} \Delta_{\rm true}^{2 - 2\beta} + \mylog{ \epsilon^{-\beta}\Delta_{\rm true}^{\beta-1}}} \mylog{\delta^{-1}}}$$ quantum circuits, where each circuit has depth proportional to $$\tilde{\mathcal{O}}( \epsilon^{-\beta} \Delta_{\rm true}^{\beta-1} 
\operatorname{log}(\eta^{-1}\Delta_{\rm true}^{1-\beta} \epsilon^{\beta-1})).$$    
\label{cor:interpolation_informal}
\end{corollary}

Note that this result interpolates between the standard limit (\( \beta = 0 \)) and the Heisenberg limit (\( \beta = 1 \)) by exchanging the circuit depth for the number of samples. Importantly, instead of scaling as \( \eta^{-2} \) as in previous methods \cite{wang2023quantum, ding2022even}, the number of repetitions here is quadratically improved to \( \eta^{-1} \) scaling. This improvement is achieved without increasing the number of quantum operations per circuit, as is typically required by standard approaches to quadratic speedups, such as amplitude amplification.
\hl{Also, while Corollary \ref{cor:interpolation_informal} pertains to the existence of algorithms, one may ask to what extent $\Delta$ can be tuned between $\epsilon$ and $\Delta_{\rm true}$ \emph{in practice}. In order to tune the algorithm parameter $\Delta$ above $\epsilon$ while keeping it below $\Delta_{\rm true}$, one needs a sufficiently accurate estimate of $\Delta_{\rm true}$. We refer the reader to Footnote \ref{ft:gap_comment} and the reference therein for the practical considerations of setting $\Delta$.}

Figure \ref{fig:intro} compares the performance of our GSEE algorithm with others in the literature. Notably, our algorithm demonstrates the best performance in terms of \( \eta \) within the regime of circuit depths between \( \tilde{\mathcal{O}}(\Delta^{-1}_{\rm true}) \) and \( \tilde{\mathcal{O}}(\epsilon^{-1}) \).

\definecolor{mygray}{gray}{0.6}

\begin{figure}
    \centering

\begin{tikzpicture}

\draw[thick=0.1, ->] (0, 0) -- (7, 0) node[below=1mm] {};
\node[] at (3, -1) {Circuit depth (${\cal T}_{\max}$)};
\draw[thick=0.1,->] (0, 0) -- (0, 5) node[left] {};
\node[rotate=90] at (-2.5, 2.5) {Total runtime};
\node[rotate=90] at (-2.0, 2) {(${\cal T}_{\rm tot}$)};
\draw (1, 0.1) -- (1, 0);
\draw (0.1, 1.7) -- (0, 1.7) node[left]{\small $\eta^{-2}\epsilon^{-1}$};
\draw (0.1, 1) -- (0, 1) node[left]{\small $\eta^{-1}\epsilon^{-1}$};
\draw (0.1, 0.5) -- (0, 0.5) node[left]{\small $\eta^{-1/2}\epsilon^{-1}$};
\node[] at (1, -0.3) {\small $\Delta_{\textup{true}}^{-1}$};
\draw (4, 0) -- (4, 0.1);
\draw (0, 3) -- (0.1, 3) node[left] {\small $\eta^{-1}\epsilon^{-2}\Delta_{\textup{true}}$};
\draw (0, 4.25) -- (0.1, 4.25) node[left] {\small $\eta^{-2}\epsilon^{-2}\Delta_{\textup{true}}$};
\node[] at (4, -0.3) {\small $\epsilon^{-1}$};
\draw (5.5, 0) -- (5.5, 0.1);
\node[] at (5.5, -0.3) {\small $\eta^{-1/2}\epsilon^{-1}$};
\draw[loosely dashed, very thick, domain=1:4] plot (\x,{8.5/(\x+1)});
\draw[violet, densely dashed, very thick, domain=1:4] plot (\x,{4.5/(\x+0.5)});

\node at (1,4.25) [circle,fill=orange,inner sep=1.5pt]{};
\node at (1,3) [circle,fill=orange,inner sep=1.5pt]{};

\node at (4,1.7) [circle,fill=orange,inner sep=1.5pt, label={[shift={(0.8,-0.4)}] }]{};
\node at (4,1) [circle,fill=orange,inner sep=1.5pt, label={[shift={(0.6,-0.4)}]}]{};
\node at (5.5,0.5) [circle,fill=orange,inner sep=1.5pt, label={[shift={(0.6,-0.4)}]}]{};

\node[] at (2.75, 2.75) {\small \cite{wang2023quantum, ding2022even}};
\node[] at (2.1, 1.25) {\small Cor. \ref{cor:interpolation_informal}};
\node at (4,1.7) [circle,fill=mygray,inner sep=1.5pt, label={[shift={(0.8,-0.4)}]\small \cite{lin2022heisenberg,wan2022randomized}}]{};
\node at (4,1) [circle,fill=mygray,inner sep=1.5pt, label={[shift={(0.6,-0.4)}]\small \cite{dong2022ground}}]{};
\node at (5.5,0.5) [circle,fill=mygray,inner sep=1.5pt, label={[shift={(0.6,-0.4)}]\small \cite{dong2022ground}}]{};

\end{tikzpicture}
    \caption{This figure shows the landscape of early fault-tolerant GSEE algorithms plotted according to their runtimes and circuit depths. The loosely dashed green curve represents our previous interpolation result \cite{wang2023quantum} (which is also implied by Ref.~\cite{ding2022even}), while the lower densely dashed violet curve depicts the result in Corollary \ref{cor:interpolation_informal}. For the purposes of this plot, we have assumed that $\epsilon<\eta\Delta_{\rm true}$, leading us to place $\eta^{-1}\epsilon^{-2}\Delta_{\textup{true}}$ above $\eta^{-2}\epsilon^{-1}$, though this assumption is not necessary for the main results of our paper.}
    \label{fig:intro}
\end{figure}

Next, we propose a novel method for performing rejection sampling on a quantum computer:

\begin{theorem}[Rejection sampling from the spectral measure]
Let \( H \) be an \( n \)-qubit Hamiltonian with spectral decomposition \( H = \sum_i E_i \ketbra{E_i}{E_i} \), and let \( \ket{\psi} \) be an \( n \)-qubit state. Define the probability measure \( p \) on \( \mathbb{R} \) as 
\begin{align}
    p(x) = \sum_i p_i \delta(x - E_i), \quad p_i = |\langle E_i | \psi \rangle|^2,
\end{align}
where \( \delta \) denotes the Dirac delta function. Furthermore, let \( \nu \) and \( \mu \) be two probability density functions on \( \mathbb{R} \) such that \( \textrm{supp}(p \ast \nu) \subset \textrm{supp}(\mu) \), where \( \ast \) denotes convolution. Finally, assume that there exists a constant \( M \in \mathbb{C} \) such that for every \( x \in \textrm{supp}(\mu) \), we have access to a \( (1, m, 0) \)-block-encoding \( U_x \) of an operator \( h(H - xI) \) which satisfies
\begin{align}
    |h(H - xI)|^2 = \frac{\nu(H - xI)}{|M|^2 \mu(x)}, \quad \|h(H - xI)\| \leq 1,
\end{align}
and that we can generate samples distributed according to \( \mu \). Then we can generate a sample distributed according to \( p \ast \nu \) from an expected \( |M|^2 \) uses of such \( U_x \)'s.
\label{thm:rejection_sampling}
\end{theorem}

So far, we have only considered algorithms for GSEE with circuit depths proportional to \( \Delta^{-1} \), where \( \Delta > \epsilon \) is a known lower bound on the spectral gap \( \Delta_{\rm true} = E_1 - E_0 \) of \( H \). However, it is desirable to develop algorithms that do not rely on this prior knowledge to accurately estimate the ground-state energy. To address this, we develop a method for certifying the ground-state energy estimate, albeit with a higher sample complexity:

\begin{theorem}[ground-state energy certification, informal] Under the same conditions as Corollary \ref{cor:interpolation_informal}, there exists an algorithm that is given, as an input, arbitrary $\sigma>0$ and $\hat{E}_0 \in [E_0-\sigma, E_0+\sigma]$ (but no lower bound on $\Delta_{\rm true}$), makes use of 
\begin{align}
\tilde{\mathcal{O}}(\eta^{-3}\epsilon^{-4}\mylog{\delta^{-1}})
\end{align}
quantum circuits of depth $\tilde{\mathcal O}(\sigma^{-1})$ and produces a new estimate $\hat{E}_0'$ of $E_0$ that it either accepts or rejects. If $|\hat{E}_0'-E_0|>\epsilon$, it rejects $\hat{E}_0'$ with probability at least $1-\delta$. If it accepts, then $|\hat{E}_0'-E_0|\leq \epsilon$ with probability of incorrect acceptance at most $\delta$.
Furthermore, the maximal $\sigma$ for which the algorithm outputs an estimate $\hat{E}_0'$ it accepts satisfies $|\hat{E}_0'-E_0|\le \epsilon$ is $\sigma=\Omega\left( {\Delta_{\rm{true}}}/{\mylog{\epsilon^{-1}\eta^{-1}}}\right)$.

\end{theorem}
Thus, with a higher sample complexity, we can also produce a certificate that the estimate is correct, even without knowing the spectral gap \( \Delta_{\rm true} \) of \( H \). Furthermore, the algorithm described above will consistently accept with high confidence if we use depths proportional to the inverse of the true spectral gap, the minimum depth required in Corollary \ref{cor:interpolation_informal}. Note, however, that in certain situations, the algorithm may still accept—and therefore provide a correct estimate—even if \( \sigma \gg \Delta_{\rm true} \). As discussed in more detail in Section \ref{sec:certifcation_sec}, this can occur, for example, when the initial state has minimal support on low-energy excitations.

We note that our certification algorithm \emph{does not} certify the spectral gap itself; rather, it certifies only our ground-state energy estimate. To the best of our knowledge, this is the first method to achieve certifiable GSEE in the regime where the circuit depth is \( \tilde{\mathcal{O}}(\Delta_{\rm true}^{-1}) \). Additionally, the assumption that we are given an estimate satisfying \( |E_0 - \hat{E}_0| \leq \sigma \) does not increase the overall circuit depth required for the certification method, as we will show how to obtain such an estimate with circuits of depth \( \tilde{\mathcal{O}}(\sigma^{-1}) \) in Section \ref{sec:gsee_alg}.

Finally, let us briefly comment on the role of \( \delta \). Our algorithm relies on empirically estimating the variance of a random variable with a certain precision. As with similar protocols, the empirical estimate will fall within a tolerated error margin with a specified failure probability. Importantly, our protocol always correctly rejects when given variance estimates that fall within the tolerated error margin. Thus, $\delta$ quantifies the probability that the estimates deviate beyond the acceptable range for verification, in which case the performance guarantee no longer holds.

\section{Low-depth ground-state energy estimation}
\label{sec:gsee_alg}
In this section, we present our low-depth quantum algorithm for ground-state energy estimation. Section~\ref{subsec:basic_idea} introduces the basic strategy and establishes its correctness. Then, Section~\ref{subsec:concrete_implementation} provides concrete implementations of the algorithm, along with an analysis of its circuit depth and runtime.

\subsection{Basic strategy}
\label{subsec:basic_idea}

Let us first introduce a basic GSEE algorithm, referred to as Algorithm \ref{alg:basic_gsee}, which generalizes the approach from Ref.~\cite{dong2022ground}. This algorithm aims to generate a coarse estimate of \( E_0 \).

To build intuition, we begin by describing a simplified and idealized version of the algorithm. Suppose that, for some \( \lambda \in (0, \pi) \), we can implement a \( (\alpha, m, 0) \)-block-encoding \( V_\lambda \) of the operator \( f_\lambda(H) \), where \( f_\lambda : [0, \pi] \to [0, 1] \) is defined such that \( f_\lambda(x) = 1 \) for \( x \in [0, \lambda] \) and \( f_\lambda(x) = 0 \) for \( x \in (\lambda, \pi] \). Let \( \ket{\psi} \) be an initial state satisfying \( |\langle \psi | E_0 \rangle|^2 \geq \eta \) for some known \( \eta \in (0, 1) \). 

If we choose \( \lambda > E_0 \), then \( \|f_\lambda(H) \ket{\psi}\|^2 \geq \eta \), as the operator \( f_\lambda(H) \) acts as the identity on the ground state space of \( H \). Conversely, if \( \lambda < E_0 \), we have \( \|f_\lambda(H) \ket{\psi}\|^2 = 0 \), since \( f_\lambda(H) \) maps all eigenspaces of \( H \) to zero.

By the definition of block-encoding, measuring the \( m \) ancilla qubits of the state \( V_\lambda \ket{0^m}\ket{\psi} \) gives outcome \( 0^m \) with probability \( \alpha^{-2} \|f_\lambda(H) \ket{\psi}\|^2 \). Thus, if such a block-encoding were available, we could infer the ground-state energy by performing a binary search to find the smallest \( \lambda \) for which \( \|f_\lambda(H) \ket{\psi}\|^2 \geq \eta \). For a given \( \lambda \), we can determine whether \( \|f_\lambda(H) \ket{\psi}\|^2 \geq \eta \) or \( \|f_\lambda(H) \ket{\psi}\|^2 = 0 \) with a success probability of \( 2/3 \) by repeating the experiment \( \mathcal{O}(\eta^{-1}) \) times and checking how often the measurement yields \( 0^m \). This approach would result in an algorithm with a runtime of \( \mytO{\eta^{-1} \mylog{\epsilon^{-1}}} \) to estimate the ground-state energy with accuracy \( \epsilon \).

However, implementing block-encodings of non-smooth functions such as \( f_\lambda \) is not feasible directly, necessitating approximations using smooth functions. This motivates the introduction of the function \( f_{a,b} \) below. Approximating \( f_\lambda \) introduces some subtleties in how we post-process the samples and conduct the binary search, but the overall strategy remains comparable to the one described above.

\begin{algorithm}[ht]
\caption{Basic GSEE algorithm}\label{alg:basic_gsee}
\begin{algorithmic}[1]
\Procedure{Basic\_GSEE}{$\epsilon$, $\delta$, $\eta$, $\alpha$, $m$} 
    \State Set $L = \lceil \logab{3/2}{\pi/\epsilon}\rceil$, $M =\lceil 11.25 \alpha^2 \eta^{-1} \myln{L \delta^{-1}} \rceil$, $\epsilon' = \min(\sqrt{0.1 \eta}, 0.05)$;
    \State $l \gets 0$, $r \gets \pi$;
    \While{$r - l > \epsilon $}
        \State $a \gets (2l+r)/3$, $b \gets (l+2r)/3$; 
        \State Let $V_{a,b}$ be a $(\alpha, m, 0)$-block-encoding of $f_{a,b}(H)$,  where $f_{a,b}: [0, \pi] \to [-1, 1]$ satisfies 
        \Statex \hspace{30pt}  $f_{a,b}(x) \ge 1-\epsilon'$ for all $x \in [l, a]$, and  $|f_{a,b}(x)| \le \epsilon'$ for all $x \in [b, \pi]$;   \label{step:Vab_block_encode}
        \State Prepare \( M \) copies of the state \( \ket{0^m}\ket{\psi} \). For each copy, apply the operation \( V_{a,b} \), then  
        \Statex \hspace{30pt}  measure the first \( m \) qubits in the standard basis. Count the number \( K \) of times the  
        \Statex \hspace{30pt} measurement outcome is \( 0^m \);
        \If{$K < 0.5 \alpha^{-2} \eta M$} 
        \State $l \gets a$; 
        \Else
        \State $r \gets b$; 
        \EndIf 
    \EndWhile
    \State \Return $[l, r]$.
\EndProcedure
\end{algorithmic}
\end{algorithm}

\begin{proposition}
Let $H$ be an $n$-qubit Hamiltonian with unknown spectral decomposition $H=\sum_{j=0}^{N-1} E_j \ket{E_j} \bra{E_j}$, where $N=2^n$, and the eigenvalues satisfy $0 \le E_0<E_1 \le E_2 \le ...\le E_{N-1} \le \pi$. Suppose we can prepare an $n$-qubit state $\ket{\psi}=\sum_{j=0}^{N-1} \gamma_j \ket{E_j}$, with an overlap with the ground state $\ket{E_0}$ satisfying $|\gamma_0|^2 \ge \eta$, where $\eta \in (0, 1)$ is known. 

Then, for any $\epsilon>0$ and $\delta \in (0, 1)$, the output of Algorithm \ref{alg:basic_gsee} is an interval of length at most $\epsilon$ that contains the ground-state energy $E_0$ with probability at least $1-\delta$. 
\label{thm:basic_gsee} 
\end{proposition}
\begin{proof}
We claim that in each iteration of the {\bf while}-loop, if $E_0 \le a=(2l+r)/3$, then $K > 0.5 \alpha^{-1} \eta M$ with probability at least $1-\delta/L$; if $E_0 \ge b=(l+2r)/3$, then $K < 0.5 \alpha^{-1} \eta M$ with probability at least $1-\delta/L$. If this is true, then after each iteration, $E_0$ remains between $l$ and $r$ with probability at least $1-\delta/L$. Since there are at most $L$ iterations in the {\bf while}-loop, $E_0$ is between the final $l$ and $r$ with probability at least $1-\delta$.

To prove the above claims, we need to bound the probability of obtaining measurement outcome $0^m$ in each iteration of the {\bf for}-loop in the cases $E_0 \le a$ and $E_0 \ge b$ respectively. If $E_0 \le a$, then this probability is 
\begin{align}
\alpha^{-2}\sum_{j=0}^{N-1} p_j f^2_{a,b}(E_j)
\ge 
\alpha^{-2} p_0 f^2_{a,b}(E_0)
\ge \alpha^{-2} \eta (1-\epsilon')^2
\ge 0.9 \alpha^{-2} \eta.
\end{align}
On the other hand, if $E_0 \ge b$, then this probability is 
\begin{align}
\alpha^{-2}\sum_{j=0}^{N-1} p_j f^2_{a,b}(E_j)
\le 
\alpha^{-2} \max_{0\le j\le N-1} f^2_{a,b}(E_j)
\le \alpha^{-2} (\epsilon')^2
\le 0.1 \alpha^{-2} \eta.
\end{align}
Next, we invoke the Chernoff-Hoeffding Theorem which states that if $X_1, X_2, \dots, X_n$ are independent and identically distributed random variables which take values in $\zo$. Let $p = \E{X_1}$. Then for any $q \in (p, 1]$ and $r \in [0, p)$, we have
\begin{align}
    \P{\frac{1}{n} \sum_{i=1}^n X_i \ge q} \le e^{-n D(q \| p)}, \\
    \P{\frac{1}{n} \sum_{i=1}^n X_i \le r} \le e^{-n D(r \| p)},    
\end{align}
where $D(x\|y) \defeq x \myln{x/y} + (1-x) \myln{(1-x)/(1-y)}$ is the Kullback–Leibler divergence between Bernoulli random variables with parameters $x$ and $y$ respectively. We also use the fact that
\begin{align}
    D(x\|y) \ge \frac{(x-y)^2}{2 \max(x, y)}, ~&~\forall x,y \in (0, 1).
\end{align}

If $E_0 \le a$, then we have $p \ge 0.9 \alpha^{-2} \eta$, and hence
\begin{align}
    \P{K/M \le 0.5 \alpha^{-2}\eta} \le e^{-M D(0.5 \alpha^{-2}\eta \| p)}
\end{align}
in which
\begin{align}
    D(0.5\alpha^{-2} \eta \| p ) \ge \frac{(p-0.5\alpha^{-2}\eta)^2}{2 p} \ge  \frac{4\eta}{45\alpha^{2}},
\end{align}
as $\frac{(p-0.5\alpha^{-2}\eta)^2}{2 p}$ increases monotonically for $p \ge 0.5\alpha^{-2}\eta$. Then by our choice of $M$, we know that $\P{K \le 0.5\alpha^{-2}\eta M} \le \delta/L$.

On the other hand, if $E_0 \ge b$, we have $p \le 0.1 \alpha^{-2} \eta$, and hence
\begin{align}
    \P{K/M \ge 0.5\alpha^{-2}\eta} \le e^{-M D(0.5\alpha^{-2}\eta \| p)}
\end{align}
in which
\begin{align}
    D(0.5\alpha^{-2}\eta \| p ) \ge \frac{(0.5\alpha^{-2}\eta-p)^2}{2 \cdot 0.5\alpha^{-2}\eta} \ge \frac{4\eta}{25\alpha^2}.
\end{align}
Then by our choice of $M$, we obtain that $\P{K \ge 0.5\alpha^{-2}\eta M} \le \delta/L$.

The theorem is thus proved.    
\end{proof}

Next, we describe the idea behind our main GSEE algorithm. We begin by using Algorithm \ref{alg:basic_gsee} to find an interval \( [l, r] \) of length \( \mathcal{O}(\Delta) \) that contains \( E_0 \) with high probability. This step requires quantum circuits of depth \( \tilde{\mathcal{O}}(1/\Delta) \). Given that the Hamiltonian \( H \) has a spectral gap of at least \( \Delta \), we know that all excited state energies are at least \( r + \myOmega{\Delta} \). 

Equipped with this knowledge, we can use techniques such as quantum signal processing (QSP) \cite{qsp_paper} and its variants \cite{dong2022ground} to filter out the unwanted spectral components in the input state \( \ket{\psi} \), obtaining a good approximation of the ground state \( \ket{E_0} \). Furthermore, the efficiency of this procedure depends on the overlap between $\ket{\psi}$ and \( \ket{E_0} \). Once we obtain a high-fidelity approximation of \( \ket{E_0} \), running quantum phase estimation (QPE) to determine the ground-state energy \( E_0 \) might seem like the next step. However, we find that this is unnecessary. Instead, performing Hadamard tests on the state allows us to infer \( E_0 \) efficiently with only a small number of samples.

Our main GSEE algorithm is formally presented in Algorithm \ref{alg:gsee}. It utilizes Algorithm \ref{alg:basic_gsee} to obtain a coarse estimate of $E_0$ and Algorithm \ref{alg:gsee_ref} to refine it. 

\begin{algorithm}[ht]
\caption{Advanced GSEE algorithm}\label{alg:gsee}
\begin{algorithmic}[1]
\Procedure{ADV\_GSEE}{$\epsilon, \delta, \Delta, \eta, \alpha_1, m_1, \alpha_2, m_2$} 
    \If{$\epsilon \ge \Delta/4$} 
        \State $[l, r] \gets \mathrm{BASIC\_GSEE}(2\epsilon, \delta, \eta, \alpha_1, m_1)$; \label{step:large_error}
        \State \Return $(l+r)/2$.
    \EndIf
    \State $[l, r] \gets \mathrm{BASIC\_GSEE}(\Delta/2, \delta/3, \eta,  \alpha_1, m_1)$;  \label{step:small_error}
    \State \Return $\mathrm{REFINE\_GSE\_ESTIMATE}(l, r, \epsilon, \delta, \Delta, \eta, \alpha_2, m_2)$. \label{step:refine_estimate}
\EndProcedure
\end{algorithmic}
\end{algorithm}

\begin{algorithm}[ht]
\caption{Refinement of ground-state energy estimate}\label{alg:gsee_ref}
\begin{algorithmic}[1]
\Procedure{REFINE\_GSE\_ESTIMATE}{$l, r, \epsilon, \delta, \Delta, \eta, \alpha_2, m_2$} 
    \State Choose constants $\zeta, \beta, \omega, c \in (0, 1)$;
    \State Set $t = \frac{4\pi \omega}{\Delta+4\epsilon}$, $\gamma = \frac{\beta \sin(t \epsilon)}{2}$, $\epsilon_1 = \mymin{\frac{\gamma \eta c}{(1 - \gamma)(1 - \eta)}, (1-\sqrt{c})^2}$, $\epsilon_2 = \frac{(1 - \beta) \sin(t \epsilon)}{\sqrt{2}}$;
    \State Set $K = \lceil 2 \myln{12/\delta} / \epsilon_2^2 \rceil$, $M = \lceil \mymax{2\alpha_2^2 K / (c \eta (1-\zeta)), 2\alpha_2^2 \myln{3/\delta} / (c \eta \zeta^2)} \rceil$;
    \If{$\eta \ge 1-\gamma$}
        \State Set $\ket{\phi}=\ket{\psi}$;
        \State Go to Step \ref{step:hadamard_test};
    \EndIf
    \State Let $V$ be a $(\alpha_2, m_2, 0)$-block-encoding of $f(H)$ where $f: [0, \pi] \to [-1, 1]$ satisfies   $f(x) \geq \sqrt{c}$  
    \Statex \hspace{15pt} for all $x \in [l, r]$, and $\lvert f(x) \rvert \leq \sqrt{\epsilon_1}$ for all $x \in [r + \Delta/2, \pi]$; \label{step:V_block_encode}
    \State Repeat the following experiment until either \( 2K \) copies of \( \ket{\phi} \coloneqq \frac{f(H) \ket{\psi}}{\|f(H) \ket{\psi}\|} \) are collected or \Statex \hspace{15pt} the experiment has been repeated \( M \) times: Prepare the state \( \ket{0^{m_2}} \ket{\psi} \), apply the operation \Statex \hspace{15pt}   \( V \) to it, and then measure the first \( m_2 \) qubits in the standard basis. Retain the state if the \Statex \hspace{15pt}  measurement outcome is \( 0^{m_2} \).
     \label{step:collect_samples}    
    \If{fewer than $2K$ copies of $\ket{\phi}$ are collected}
        \State \textbf{abort}.
    \EndIf
    \State Run the circuit in Figure \ref{fig:quantum_circuits} (a) with initial state $\ket{\phi}$, setting $\tau = t$, and $W = I$ or $S^\dagger$. \Statex \hspace{15pt} Perform $K$ iterations for each case, and let $\hat{X}$ or $\hat{Y}$ be the average of the $K$ measurement \Statex\hspace{15pt} outcomes for $W=I$ or $S^\dagger$,
     respectively; \label{step:hadamard_test}
    \State Find the unique $\tilde{E}_0 \in [l-\epsilon, r+\epsilon]$ such that $\tilde{E}_0 t  \equiv \myArg{\hat{X} + i\hat{Y}} (\operatorname{mod}~2\pi)$;
    \State \Return $\operatorname{clip}(\tilde{E}_0, l, r)$.
\EndProcedure
\end{algorithmic}
\end{algorithm}

\begin{proposition}
\label{thm:main_Gaussian}
Let $H$ be an $n$-qubit Hamiltonian with unknown spectral decomposition $H=\sum_{j=0}^{N-1} E_j \ket{E_j} \bra{E_j}$, where $N=2^n$, and the eigenvalues satisfy $0 \le E_0<E_1 \le E_2 \le ...\le E_{N-1} \le \pi$. Suppose we know a lower bound $\Delta>0$ on the spectral gap $E_1-E_0$ of $H$. Additionally, suppose we can prepare an $n$-qubit state $\ket{\psi}=\sum_{j=0}^{N-1} \gamma_j \ket{E_j}$, with an overlap with the ground state $\ket{E_0}$ satisfying $|\gamma_0|^2 \ge \eta$, where $\eta \in (0, 1)$ is known. 

Then, for any $\epsilon>0$ and $\delta \in (0, 1)$, the output of Algorithm \ref{alg:gsee} is an estimate $\hat{E_0}$ of the ground-state energy $E_0$ such that $|E_0-\hat{E}_0|\le \epsilon$ with probability at least $1-\delta$. 
\label{thm:main}
\end{proposition}

\begin{proof}
If $\epsilon \ge \Delta/4$, by Theorem \ref{thm:basic_gsee}, $[l,r]$ at step \ref{step:large_error} of Algorithm \ref{alg:gsee} has length at most $2\epsilon$ and contains $E_0$ with probability at least $1-\delta$. Therefore, $(l+r)/2$ is $\epsilon$-close to $E_0$ with probability at least $1-\delta$, as desired.

We now consider the case where $\epsilon < \Delta/4$. By Theorem \ref{thm:basic_gsee}, $[l,r]$ at step \ref{step:small_error} of Algorithm \ref{alg:gsee} has length at most $\Delta/2$ and contains $E_0$ with probability at least $1-\delta/3$. Assuming this event occurs, we get $r+\Delta/2 \le  E_0 + \Delta \le E_1$. Thus, $E_0$ lies within $[l,r]$, while $E_1, E_2, \dots, E_{N-1}$ lie within $[r+\Delta/2, \pi]$. As a result, by the properties of $f$ at step \ref{step:V_block_encode} of Algorithm \ref{alg:gsee_ref}, we know that $f(E_0) \ge \sqrt{c}$, and $|f(E_j)|\le \sqrt{\epsilon_1}$ for all $j\ge 1$. 

If $\eta \ge 1-\gamma$, then $\ket{\phi}=\ket{\psi}$ and it satisfies 
$|\braket{E_0}{\phi}|^2 \ge \eta \ge 1-\gamma$. Otherwise, $\ket{\phi}=\frac{f(H)\ket{\psi}}{\norm{f(H)\ket{\psi}}}$. Given that $\ket{\psi}=\sum_{j=0}^{N-1} \gamma_j \ket{E_j}$ where $|\gamma_0|^2 \ge \eta$, we get 
\begin{align}
    |\braket{E_0}{\phi}|^2 &= \frac{|\bra{E_0}f(H)\ket{\psi}|^2}{\|f(H)\ket{\psi}\|^2}  \\
    &=\frac{f^2(E_0) |\gamma_0|^2}{
    f^2(E_0) |\gamma_0|^2 + \sum_{j=1}^{N-1}f^2(E_j) |\gamma_j|^2} \\
    &\ge 
    \frac{\eta c}{\eta c + (1-\eta) \epsilon_1} \\
   &\ge 1-\gamma,
\end{align}
where the last step follows from our choice of $\epsilon_1$. Therefore, we  always have $|\braket{E_0}{\phi}|^2 \ge 1-\gamma$ regardless of the value of $\eta$.

Let us write $\ket{\phi}=\sum_{j=0}^{N-1} \xi_j \ket{E_j}$ for some $\xi_j$'s in which $|\xi_0|^2 \ge 1-\gamma$. Then we claim that $\bra{\phi} e^{\i Ht} \ket{\phi}$ is close to $e^{\i E_0t}$. Precisely,
\begin{align}
    \abs{\bra{\phi} e^{\i Ht} \ket{\phi} - e^{\i E_0t}} =\abs{\sum_{j=1}^{N-1} |\xi_j|^2 e^{\i  E_jt} + (|\xi_0|^2-1) e^{\i E_jt}}
    \le \sum_{j=1}^{N-1} |\xi_j|^2  + ||\xi_0|^2-1| 
    \le 2\gamma.
\end{align}

When we run the circuit in Figure \ref{fig:quantum_circuits} (a) with initial state $\phi$, setting $\tau=t$ and $W=I$ or $S^\dagger$, the measurement outcome $d \in \lrcb{1, -1}$ has probability distribution
\begin{align}
    \P{d=\pm 1|W=I} = \frac{1 \pm \realpart{\bra{\phi} e^{\i Ht} \ket{\phi}}}{2}, \\
    \P{d=\pm 1|W=S^\dagger} = \frac{1 \pm \imaginarypart{\bra{\phi} e^{\i Ht} \ket{\phi}}}{2},
\end{align}
respectively. It follows that
\begin{align}
\E{d|W=I} =\realpart{\bra{\phi} e^{\i Ht} \ket{\phi}},\\
\E{d|W=S^\dagger} =\imaginarypart{\bra{\phi} e^{\i Ht} \ket{\phi}}.
\end{align}

Since $\hat{X}$ and $\hat{Y}$ are the averages of $K$ independent and identically distributed $\pm 1$ variables with means $\realpart{\bra{\phi} e^{iHt} \ket{\phi}}$ and $\imaginarypart{\bra{\phi} e^{iHt} \ket{\phi}}$, respectively, by Hoeffding’s inequality and our choice of $K$, we get 
\begin{align}
    \P{\abs{\hat{X}-\realpart{\bra{\phi} e^{\i Ht} \ket{\phi}}}\le \epsilon_2} \ge 1- 2 e^{-K\epsilon_2^2/2}
    \ge 1 - \frac{\delta}{6},     \label{eq:accurate_estimate_x}    \\
    \P{\abs{\hat{Y}-\imaginarypart{\bra{\phi} e^{\i Ht} \ket{\phi}}}\le \epsilon_2}  
    \ge 1 - 2 e^{-K\epsilon_2^2/2}
    \ge 1 - \frac{\delta}{6}.
    \label{eq:accurate_estimate_y}
\end{align}
Assuming these event occurs, we have
\begin{align}
    \abs{\hat{X}+\i \hat{Y} - e^{\i E_0 t}}
\le 
    \abs{\hat{X}+\i \hat{Y} - \bra{\phi} e^{iHt} \ket{\phi}} +
    \abs{\bra{\phi} e^{iHt} \ket{\phi}- e^{\i E_0 t}}
    \le \sqrt{2}\epsilon_2 + 2\gamma 
    = \mysin{t \epsilon},
\end{align}
where the last step follows from our choice of $\epsilon_2$ and $\gamma$. This implies that the argument of $\hat{X}+\i \hat{Y}$ is $t \epsilon$-close to $E_0 t$ modulo $2\pi$. Then since $E_0 \in [l, r]$ and $(r-l+2\epsilon)t \le 2\pi \omega < 2\pi$, there exists a unique $\tilde{E}_0 \in [l-\epsilon, r+\epsilon]$ such that
$\tilde{E}_0 t  \equiv \myArg{\hat{X} + i\hat{Y}} (\operatorname{mod}~2\pi)$. This $\tilde{E}_0$ satisfies $|\tilde{E}_0t-E_0t|\le t\epsilon$ and hence  $|\tilde{E}_0-E_0|\le \epsilon$. Since $E_0 \in [l, r]$, we can clip $\tilde{E}_0$ to $[l, r]$ and it will still be $\epsilon$-close to $E_0$.

It remains to prove that the algorithm succeeds with probability at least $1-\delta$. The algorithm fails only if one of the following events happens:
\begin{itemize}
    \item $[l,r]$ at step \ref{step:small_error} of Algorithm \ref{alg:gsee} does not contain $E_0$;
    \item Few than $2K$ copies of $\ket{\phi}$ are collected at step \ref{step:collect_samples} of Algorithm \ref{alg:gsee_ref}, if $\eta<1-\gamma$;
    \item $\hat{X}$ is $\epsilon_2$-far from $\realpart{\bra{\phi} e^{\i Ht} \ket{\phi}}$, or $\hat{Y}$ is $\epsilon_2$-far from $\imaginarypart{\bra{\phi} e^{\i Ht} \ket{\phi}}$.
\end{itemize}
The first event occurs with probability at most $\delta/3$ by Theorem \ref{thm:basic_gsee}, and the third event occurs with probability at most $\delta/3$ by Eqs.\eqref{eq:accurate_estimate_x} and \eqref{eq:accurate_estimate_y}. We will show below that the second event occurs with probability at most $\delta/3$. Therefore, the total probability of the algorithm failing is at most $\delta$.

Since $V$ is a $(\alpha_2,m_2,0)$-block-encoding of $f(H)$, we have
\begin{align}
    V\ket{0^{m_2}}\ket{\psi}
    =\alpha_2^{-1}\ket{0^{m_2}}f(H)\ket{\psi}
    + \sum_{j \in \zo^{m_2},~j \neq 0^{m_2}} \ket{j}\ket{\phi_j},
\end{align}
where $\ket{\phi_j}$'s are unnormalized $n$-qubit states. Therefore, when measuring the first $m$ qubits of this state, the probability of obtaining measurement outcome $0^m$ is
\begin{align}
    \mu \defeq \|\bra{0^{m_2}} V\ket{0^{m_2}}\ket{\psi}\|^2
    =\alpha_2^{-2}\norm{f(H)\ket{\psi}}^2
    =\alpha_2^{-2}\sum_{j=0}^{N-1} |\gamma_j|^2 f^2(E_j)
    \ge \alpha_2^{-2} |\gamma_0|^2 f^2(E_0)
    \ge \alpha_2^{-2} \eta c.
\end{align}
Consequently, the probability of keeping each state at step \ref{step:collect_samples} of Algorithm \ref{alg:gsee_ref} is at least $\alpha_2^{-2} \eta c$. Let $A$ denote the number of collected states after $M$ repetitions of the experiment at this step. By the multiplicative Chernoff bound, we obtain
\begin{align}
    \P{A \ge \mu(1-\zeta) M} 
    \ge 1-e^{-\zeta^2 \mu M/2}.
\end{align}
Our choice of $M$ ensures that $\mu(1-\zeta) M \ge 2K$
and $e^{-\zeta^2 \mu M/2} \le \delta/3$. Therefore, we have
\begin{align}
    \P{A \ge 2K} 
    \ge 1-\delta/3,
\end{align}
as claimed.

The theorem is thus proved.

\end{proof}

Note that in Algorithm \ref{alg:gsee_ref}, $f(H)$ is not necessarily close to a projection operator. Since there is a single eigenvalue of $H$, namely the ground-state energy $E_0$, that lies in $[l,r]$, we only need to ensure that $f(x)=\myOmega{1}$ for all $x\in [l,r]$. It is not necessary to have $f(x) \approx 1$ to for all such values of $x$. This flexibility in the construction of $f$ helps reduce the cost of implementing $f(H)$.

Moreover, the parameters \( \zeta, \beta, \omega, c \in (0,1) \) can be adjusted to optimize the algorithm's efficiency. These parameters are introduced to balance the costs of different components of the algorithm. For instance, choosing a smaller \( \omega \) results in a smaller \( t \), leading to shorter circuit depths. However, this also reduces \( \gamma \), \( \epsilon_1 \), and \( \epsilon_2 \), making the implementation of \( f(H) \) more costly and requiring more repetitions in step \ref{step:collect_samples}. Thus, there is a trade-off between circuit depth and overall runtime. 

Similarly, picking a smaller \( c\) can reduce the cost of implementing \( f(H) \), but it also increases the number of repetitions needed in step \ref{step:collect_samples}. This demonstrates the need to carefully balance these parameters to achieve optimal performance.

\subsection{Concrete implementation}
\label{subsec:concrete_implementation}
Next, we provide a concrete implementation of the operations $V_{a,b}$ in Algorithm \ref{alg:basic_gsee} and $V$ in Algorithm \ref{alg:gsee_ref} by utilizing the technique of \emph{quantum eigenvalue transformation of unitary matrices with real polynomials} (QET-U) \cite{dong2022ground}. This technique enables the implementation of certain low-degree trigonometric polynomials of $H/2$ using a small number of applications of controlled-$e^{\i H}$ and controlled-$e^{-\i H}$. Specifically, Theorem 1 of Ref. \cite{dong2022ground} states that:

\begin{lemma}
For any even real polynomial $F(x)$ of degree $d$ satisfying $|F(x)|\le 1$ for all $x \in [-1, 1]$, we can implement a $(1, 1, 0)$-block-encoding of $F(\cos(H/2))$ using $\mathcal{O}(d)$ queries of controlled-$e^{\i H}$ or controlled-$e^{-\i H}$, along with $\mathcal{O}(d)$ additional elementary quantum gates.
\label{lem:qetu}
\end{lemma}

Note that a function $f(x)$ can be expressed as $F(\mycos{x/2})$ for 
an even real polynomial $F$ of degree $d$ if and only if 
$f(x)=\sum_{j=0}^{d/2} \alpha_j\mycos{jx}$    
for some coefficients $\alpha_j \in \R$. It is sometimes more convenient to work with the expression $\sum_{j=0}^{d/2} \alpha_j\mycos{jx}$ than with $F(\mycos{x/2})$. Moreover, for any $\xi, t \in \R$, we can implement controlled-$e^{\i(H-\xi I)t}$ using one controlled-$e^{\i Ht}$ and one $Z$-rotation gate on the control qubit. Combining these facts and Lemma \ref{lem:qetu} yields:

\begin{corollary}
Let $d \in \N$ and $\alpha_0, \alpha_1, \dots, \alpha_d \in \R$ be such that
$G(x) \defeq \sum_{j=0}^d \alpha_j\mycos{jx}$ satisfies $|G(x)| \le 1$ for all $x \in \R$. Then for any $\xi \in \R$, we can implement a $(1, 1, 0)$-block-encoding of $G(H-\xi I)$ using $\mathcal{O}(d)$ queries of controlled-$e^{\i H}$ or controlled-$e^{-\i H}$, along with $\mathcal{O}(d)$ additional elementary quantum gates.
\label{cor:qetu}
\end{corollary}

To employ QET-U to implement the block-encodings of $f_{a,b}(H)$ in Algorithm \ref{alg:basic_gsee} and $f(H)$ in Algorithm \ref{alg:gsee_ref}, we need to find suitable trigonometric approximations of these functions. To this end, we prove the following:

\begin{lemma}
For any \(z \in (0, \pi) \),  \( \delta \in (0, \pi-z] \) and \(\epsilon \in (0, 1) \), there exists an 
even function $F_{\delta, \epsilon}(x) = \sum_{j=0}^d \beta_{\delta, \epsilon, j}\mycos{jx}$ for some integer \(d = \mathcal{O}\left(\delta^{-1} \mylog{\delta^{-1}\epsilon^{-1}}\right) \) and efficiently-computable  
$\beta_{\delta,\epsilon,0}, \beta_{\delta,\epsilon,1}, \dots, \beta_{\delta,\epsilon,d} \in \R$, such that \( |F_{\delta, \epsilon}(x)| \leq 1 \) for all \( x \in \R \). Furthermore, the function \( G_{z;\delta, \epsilon}(x) \coloneqq F_{\delta, \epsilon}(x - z-\delta/2+\pi/2) \) meets the following conditions:
\begin{itemize}
    \item \( 1 - \epsilon \leq G_{z;\delta, \epsilon}(x) \leq 1 \), for all \( x \in [0, z] \);
    \item \( |G_{z;\delta, \epsilon}(x)| \leq \epsilon \), for all \( x \in [z + \delta, \pi] \).
\end{itemize}
\label{lem:fourier_approx_threshold}
\end{lemma}

\begin{proof}
By Lemma 6 of Ref.~\cite{lin2022heisenberg} and its proof, we know that for every $\delta_1 \in (0, \delta^*)$, where $\delta^*=2\arctan(1-1/\sqrt{2})\approx 0.5698$, and $\epsilon_1 \in (0, 1)$, there exists an odd integer $d_1=\mathcal{O}\lrb{\delta_1^{-1} \mylog{\delta_1^{-1}\epsilon_1^{-1}}}$ and
an odd, $2\pi$-periodic function $Q_{\delta_1, \epsilon_1}(x)$ of the form
\begin{align}
Q_{\delta_1, \epsilon_1}(x)=\sum_{j=0}^{(d_1-1)/2} \alpha_{\delta_1,\epsilon_1,j} \mysin{(2j+1)x}
\end{align}
such that $Q_{\delta_1, \epsilon_1}$ is $\epsilon_1$-close to the $2\pi$-periodic Heaviside function i.e. $H(x)=1$ if $x \in [2k\pi, (2k+1)\pi)$, and $0$ if $x \in [(2k-1)\pi, 2k\pi)$, for $k \in \Z$, on the interval $[-\pi+\delta_1, -\delta_1] \cup [\delta_1, \pi-\delta_1]$. That is, $|Q_{\delta_1,\epsilon_1}(x)-H(x)|\le \epsilon_1$ for all $x\in [-\pi+\delta_1, -\delta_1] \cup [\delta_1, \pi-\delta_1]$. Furthermore, $-\epsilon_1/2 \le Q_{\delta_1,\epsilon_1}(x) \le 1+\epsilon_1$ for all $x\in \R$. Finally, the coefficients $\alpha_{\delta_1, \epsilon_1, j}$'s are real and efficiently computable.

Now we set $\delta_1=\mymin{\delta/2, \delta^*}$ and $\epsilon_1=\epsilon/2$, and define 
\begin{align}
F_{\delta,\epsilon}(x) &\defeq \frac{1-(1-\epsilon_1)Q_{\delta_1,\epsilon_1}(x-\pi/2)}{2} \\   
&=\frac{1+(1-\epsilon_1)\sum_{j=0}^{(d_1-1)/2}(-1)^j\alpha_{\delta_1,\epsilon_1,j}\mycos{(2j+1)x}}{2}.
\end{align}

We claim that $F_{\delta,\epsilon}$ satisfies all the desired properties. First, $F_{\delta, \epsilon}(x)$ is an even function and can be expressed as $ \sum_{j=0}^d \beta_{\delta, \epsilon, j}\mycos{jx}$ for some $d=\mathcal{O}\lrb{\delta^{-1} \mylog{\delta^{-1}\epsilon^{-1}}}$ and efficiently-computable  
$\beta_{\delta,\epsilon,0}$, $\beta_{\delta,\epsilon,1}$, $\dots$, $\beta_{\delta,\epsilon,d} \in \R$. 

In addition, since $-\epsilon_1/2 \le Q_{\delta_1,\epsilon_1}(x-\pi/2) \le 1+\epsilon_1$ and $\epsilon_1=\epsilon/2$, we obtain $|(1-\epsilon_1)Q_{\delta_1,\epsilon_1}(x-\pi/2)| \le 1$, and thus  $|F_{\delta,\epsilon}(x)|\le 1$, for all $x \in \R$.

Finally, since $|Q_{\delta_1,\epsilon_1}(x)-H(x)|\le \epsilon_1$ for all $x\in [-\pi+\delta_1, -\delta_1] \cup [\delta_1, \pi-\delta_1]$ and $\epsilon_1=\epsilon/2$, we get that:
\begin{itemize}
    \item $1-\epsilon \le F_{\delta,\epsilon}(x) \le 1$, for all $x \in [(2k-1/2)\pi+\delta_1, (2k+1/2)\pi-\delta_1]$;
    \item $|F_{\delta,\epsilon}(x)|\le \epsilon$, for all $x \in [(2k+1/2)\pi+\delta_1, (2k+3/2)\pi-\delta_1]$,
\end{itemize}
for any $k \in \Z$. Now given that  \( 0 < z < z + \delta \le \pi \), $\delta_1 \le \delta/2$, and \( G_{z;\delta, \epsilon}(x) = F_{\delta, \epsilon}(x - z - \delta/2+\pi/2) \), we obtain:
\begin{itemize}
     \item \( 1 - \epsilon \leq G_{z;\delta, \epsilon}(x) \leq 1 \), for all \( x \in [0, z] \subseteq  [z+\delta/2+\delta_1-\pi, z+\delta/2-\delta_1] \);
    \item \( |G_{z;\delta, \epsilon}(x)| \leq \epsilon \), for all \( x \in [z + \delta, \pi] \subseteq [z+\delta/2+\delta_1, z+\delta/2-\delta_1+\pi] \),
\end{itemize}
as claimed.
\end{proof}

It follows from Corollary \ref{cor:qetu} and Lemma  \ref{lem:fourier_approx_threshold} that:

\begin{corollary}
For any \(z \in (0, \pi) \),  \( \delta \in (0, \pi-z] \) and \(\epsilon \in (0, 1) \), let $G_{z; \delta,\epsilon}(x)$ be the function in Lemma \ref{lem:fourier_approx_threshold}. Then we can implement a \( (1, 1, 0) \)-block-encoding of \( G_{z;\delta,\epsilon}(H) \) using $\mathcal{O}\lrb{\delta^{-1} \mylog{\delta^{-1}\epsilon^{-1}}}$ queries of controlled-$e^{\i H}$ or controlled-$e^{-\i H}$, along with $\mathcal{O}\lrb{\delta^{-1} \mylog{\delta^{-1}\epsilon^{-1}}}$ additional elementary quantum gates.
\label{cor:qetu_implement_threshold}
\end{corollary}

With the concrete implementation of all components of Algorithm \ref{alg:gsee} in place, we are now ready to analyze its complexity:

\begin{theorem}
Let $H$ be an $n$-qubit Hamiltonian with unknown spectral decomposition $H=\sum_{j=0}^{N-1} E_j \ket{E_j} \bra{E_j}$, where $N=2^n$, and the eigenvalues satisfy $0 \le E_0<E_1 \le E_2 \le ...\le E_{N-1} \le \pi$. Suppose we know a lower bound $\Delta>0$ on the spectral gap $E_1-E_0$ of $H$. Additionally, suppose we can prepare an $n$-qubit state $\ket{\psi}=\sum_{j=0}^{N-1} \gamma_j \ket{E_j}$, with an overlap with the ground state $\ket{E_0}$ satisfying $|\gamma_0|^2 \ge \eta$, where $\eta \in (0, 1)$ is known. 

Then there exists a quantum algorithm that, for any given $\epsilon>0$ and $\delta \in (0, 1)$, estimates $E_0$ within additive error $\epsilon$ with probability at least $1-\delta$ by making use of $$\mytO{\eta^{-1} \lrb{\Delta^{2}\epsilon^{-2} + \mylog{\Delta^{-1}}}\mylog{\delta^{-1}}}$$ quantum circuits. Each circuit involves multiple applications of controlled-\( e^{\i Ht_j} \) for various \( t_j \) such that $$\sum_j |t_j| = \tilde{\mathcal{O}}(\Delta^{-1} \mylog{\eta^{-1} \Delta \epsilon^{-1}}),$$ and uses $$\mytO{\Delta^{-1} \operatorname{log}(\eta^{-1}\Delta \epsilon^{-1})}$$ additional elementary quantum gates.
\label{thm:main_in_he}
\end{theorem}
\begin{proof}
We run Algorithm \ref{alg:gsee} to estimate the ground-state energy \( E_0 \) of the Hamiltonian \( H \). This algorithm requires calling Algorithm \ref{alg:basic_gsee} at step \ref{step:large_error} or \ref{step:small_error}. In Algorithm \ref{alg:basic_gsee}, we need to find a function \( f_{a,b}: [0, \pi] \to [-1, 1] \) such that \( f_{a,b}(x) \ge 1 - \epsilon' \) for all \( x \in [l, a] \) and \( \lvert f_{a,b}(x) \rvert \le \epsilon' \) for all \( x \in [b, \pi] \). Lemma \ref{lem:fourier_approx_threshold} provides a trigonometric polynomial \( f_{a,b} \) that satisfies these properties. Then we can apply Corollary \ref{cor:qetu_implement_threshold} to implement a \( (1,1,0) \)-block-encoding of \( f_{a,b}(H) \). Thus, we set \( \alpha_1 = 1 \) and \( m_1 = 1 \) in Algorithm \ref{alg:gsee}.

Meanwhile, if \( \epsilon < \Delta/4 \), we need to invoke Algorithm \ref{alg:gsee_ref} at step \ref{step:refine_estimate}. In this subroutine, we need to find a function \( f: [0, \pi] \to [-1, 1] \) such that \( f(x) \geq \sqrt{c} \) for all \( x \in [l, r] \), and \( \lvert f(x) \rvert \leq \sqrt{\epsilon_1} \) for all \( x \in [r + \Delta/2, \pi] \), provided that \( \eta < 1 - \gamma \). Note that our choice of \( \epsilon_1 \) ensures that \( \sqrt{c} \le 1 - \sqrt{\epsilon_1} \). Thus, it suffices to find a trigonometric polynomial \( f \) such that \( f(x) \ge 1 - \sqrt{\epsilon_1} \) for all \( x \in [l, r] \), and \( \lvert f(x) \rvert \le \sqrt{\epsilon_1} \) for all \( x \in [r + \Delta/2, \pi] \). This can be accomplished using Lemma \ref{lem:fourier_approx_threshold} again. Then we apply Corollary \ref{cor:qetu_implement_threshold} to implement a \( (1,1,0) \)-block-encoding of \( f(H) \). Thus, we set \( \alpha_2 = 1 \) and \( m_2 = 1 \) in Algorithm \ref{alg:gsee}.

By Proposition \ref{thm:main}, Algorithm \ref{alg:gsee} outputs a value that is \( \epsilon \)-close to \( E_0 \) with probability at least \( 1 - \delta \), as desired.

It remains to analyze the cost of Algorithm \ref{alg:gsee}, which depends on the individual costs of Algorithms \ref{alg:basic_gsee} and \ref{alg:gsee_ref}. We will examine each of these algorithms separately.

In BASIC$\_$GSEE($\tilde{\epsilon}$, $\tilde{\delta}$, $\eta$, $1$, $1$), we set $L=\myO{\mylog{\tilde{\epsilon}^{-1}}}$, $M=\myO{\eta^{-1}\mylog{\tilde{\delta}^{-1}\mylog{\tilde{\epsilon}^{-1}}}}$ and $\epsilon'=\Theta\lrb{\sqrt{\eta}}$. Since we have $b-a=\myOmega{\tilde{\epsilon}}$ at step \ref{step:Vab_block_encode}, by Corollary \ref{cor:qetu_implement_threshold}, the implementation of $V_{a,b}$ takes $$\mytO{\tilde{\epsilon}^{-1} \mylog{\epsilon'^{-1}}}
=\mytO{\tilde{\epsilon}^{-1} \mylog{\eta^{-1}}}$$ queries of controlled-$e^{\i H}$ or controlled-$e^{-\i H}$, and $\mytO{\tilde{\epsilon}^{-1} \mylog{\eta^{-1}}}$ additional elementary quantum gates. We perform a total of 
\begin{align}
\myO{ML}=\mytO{\eta^{-1}\mylog{\tilde{\epsilon}^{-1}}\mylog{\tilde{\delta}^{-1}}} 
\end{align}
such operations in this algorithm. 

Note that we always set $\tilde{\epsilon}=\myOmega{\Delta}$ and $\tilde{\delta}=\myTheta{\delta}$ whenever  BASIC$\_$GSEE($\tilde{\epsilon}$, $\tilde{\delta}$, $\eta$, $1$, $1$) is called within Algorithm \ref{alg:gsee}. Therefore, this part of Algorithm \ref{alg:gsee} requires at most $$\mytO{\eta^{-1}\mylog{\Delta^{-1}}\mylog{ \delta^{-1}}}$$ quantum circuits, where each circuit uses
$$\mytO{ \Delta^{-1} \mylog{\eta^{-1}}}$$ queries of controlled-$e^{\i H}$ or controlled-$e^{-\i H}$, and $$\mytO{\Delta^{-1} \mylog{\eta^{-1}}}$$ additional elementary quantum gates.

Next, we analyze the cost of $\mathrm{REFINE\_GSE\_ESTIMATE}(l, r, \epsilon, \delta, \Delta, \eta, 1, 1)$, assuming $\epsilon<\Delta/4$. By our choice of parameters, we have $t=\myTheta{\Delta^{-1}}$, $\gamma=\myTheta{\epsilon \Delta^{-1}}$ \footnote{Note that $0<t\epsilon<4\pi \epsilon/(\Delta+4\epsilon) < \pi/2$, and hence $ 2 t\epsilon / \pi \le \mysin{t\epsilon} \le t\epsilon$.}, $\epsilon_1=\myTheta{\eta \gamma}=\myTheta{\eta \epsilon \Delta^{-1}}$, $\epsilon_2=\myTheta{\epsilon \Delta^{-1}}$, 
\begin{align}
K=\myO{\epsilon_2^{-2} \mylog{\delta^{-1}} }=\myO{\epsilon^{-2}\Delta^2 \mylog{\delta^{-1}}},    
\end{align}
and
\begin{align}
M=\myO{\mymax{K, \mylog{\delta^{-1}}}\cdot \eta^{-1}}=\myO{\eta^{-1}\epsilon^{-2}\Delta^2\mylog{\delta^{-1}}}.
\end{align}
By Corollary \ref{cor:qetu_implement_threshold}, the operation $V$ at step \ref{step:V_block_encode} can be implemented with  
$$\mytO{\Delta^{-1} \mylog{\epsilon_1^{-1}}}=\mytO{\Delta^{-1} \mylog{\eta^{-1}\Delta\epsilon^{-1}}}$$ 
queries of controlled-$e^{\i H}$ or controlled-$e^{-\i H}$, and $\mytO{\Delta^{-1} \mylog{\eta^{-1}\Delta\epsilon^{-1}}}$ additional elementary quantum gates. We perform at most 
\begin{align}
M=\myO{\eta^{-1}\epsilon^{-2}\Delta^2\mylog{\delta^{-1}}}
\end{align}
such operations in the algorithm. 

In addition, we also need to perform $2K$ Hadamard tests with an evolution time $t=\myTheta{\Delta^{-1}}$ on the collected states from step \ref{step:collect_samples}. This increases the maximum evolution time 
 of our circuits by $\myTheta{\Delta^{-1}}$, which is insignificant compared to the original evolution time of $\mytO{\Delta^{-1} \mylog{\eta^{-1}\Delta\epsilon^{-1}}}$.

Overall, when we account for the costs of all components of   ADV$\_$GSEE($\epsilon, \delta, \Delta, \eta, 1, 1, 1, 1)$, we conclude that
it requires
 $$\mytO{\eta^{-1} \lrb{\Delta^{2}\epsilon^{-2} + \mylog{\Delta^{-1}}}\mylog{\delta^{-1}}}$$ quantum circuits, where each circuit evolves $H$ up to a time of $$\tilde{\mathcal{O}}(\Delta^{-1} \mylog{\eta^{-1} \Delta \epsilon^{-1}}),$$ and uses $$\mytO{\Delta^{-1} \operatorname{log}(\eta^{-1}\Delta \epsilon^{-1})}$$ additional elementary quantum gates, as claimed.
\end{proof}

Although Lemma \ref{lem:fourier_approx_threshold} provides explicit  constructions for the functions \( f_{a,b} \) in Algorithm \ref{alg:basic_gsee} and \( f \) in Algorithm \ref{alg:gsee_ref}, these trigonometric polynomials may have unnecessarily high degrees. Alternatively, we can determine the lowest-degree trigonometric polynomial that meets the desired conditions by solving a sequence of linear programming problems, as detailed in Appendix \ref{sec:minimize_degree}. This approach helps to minimize the costs associated with implementing the operations \( V_{a,b} \) in Algorithm \ref{alg:basic_gsee} and \( V \) in Algorithm \ref{alg:gsee_ref}.

We emphasize that $\Delta$ in Theorem \ref{thm:main_in_he} is merely a lower bound on the true spectral gap $\Delta_{\rm true} \defeq E_1-E_0$ of the Hamiltonian $H$. By tuning it between $\Theta(\epsilon)$ and $\Theta(\Delta_{\rm true})$, we obtain a class of GSEE algorithms that smoothly transition from having low circuit depth to having Heisenberg-limit scaling of runtime with respect to $\epsilon$:

\begin{corollary}

Let $H$ be an $n$-qubit Hamiltonian with unknown spectral decomposition $H=\sum_{j=0}^{N-1} E_j \ket{E_j} \bra{E_j}$, where $N=2^n$ and the eigenvalues satisfy $0 \le E_0<E_1 \le E_2 \le ...\le E_{N-1} \le \pi$,
and let $\Delta_{\rm true} \defeq E_1-E_0$ be the spectral gap of $H$. Additionally, suppose we can prepare an $n$-qubit state $\ket{\psi}=\sum_{j=0}^{N-1} \gamma_j \ket{E_j}$, with an overlap with the ground state $\ket{E_0}$ satisfying $|\gamma_0|^2 \ge \eta$, where $\eta \in (0, 1)$ is known. 

Then for every $\beta \in [0, 1]$, there exists an algorithm that, for any given $\epsilon>0$ and $\delta\in (0,1)$, estimates $E_0$ within additive error $\epsilon$ with probability at least $1-\delta$ by making use of 
$$\tilde{\mathcal{O}}\lrb{\eta^{-1} \lrb{ \epsilon^{-2+2\beta} \Delta_{\rm true}^{2 - 2\beta} + \mylog{\epsilon^{-\beta} \Delta_{\rm true}^{\beta-1}}} \mylog{\delta^{-1}}}$$ 
quantum circuits. Each circuit involves multiple applications of controlled-\( e^{\i Ht_j} \) for various \( t_j \) such that 
$$ \sum_j |t_j| = \tilde{\mathcal{O}}(\epsilon^{-\beta} \Delta_{\rm true}^{-1+\beta} 
\operatorname{log}(\eta^{-1}\Delta_{\rm true}^{1-\beta} \epsilon^{\beta-1})),$$ and uses 
$$\tilde{\mathcal{O}}(\epsilon^{-\beta} \Delta_{\rm true}^{-1+\beta}
\operatorname{log}(\eta^{-1}\Delta_{\rm true}^{1-\beta} \epsilon^{\beta-1}))$$
additional elementary quantum gates.
\label{cor:interpolation_he}
\end{corollary}

\begin{proof}
This corollary is the direct consequence of setting $\Delta=\Theta(\epsilon^{\beta} \Delta_{\rm true}^{1-\beta})$ for $\beta \in [0, 1]$ in Theorem \ref{thm:main_in_he}.
\end{proof}

In particular, setting $\beta=0$ or $1$ yields:
\begin{itemize}
    \item $\Delta = \Theta(\Delta_{\rm true})$, for which we get an algorithm that
    makes use of $\mytO{\eta^{-1} \epsilon^{-2} \Delta_{\rm true}^{2}\mylog{\delta^{-1}}}$ quantum circuits, where each circuit evolves $H$ for a total time of $\mytO{\Delta_{\rm true}^{-1}\operatorname{log}(\eta^{-1} \Delta_{\rm true} \epsilon^{-1})}$, and uses $\mytO{\Delta_{\rm true}^{-1}\operatorname{log}(\eta^{-1} \Delta_{\rm true}\epsilon^{-1})}$ extra elementary quantum gates; or
    \item $\Delta = \Theta(\epsilon)$, for which we get an algorithm that makes use of $\mytO{\eta^{-1} \mylog{\epsilon^{-1}} \mylog{\delta^{-1}}}$ quantum circuits, where each circuit evolves $H$ for a total time of  $\mytO{\epsilon^{-1}\operatorname{log}(\eta^{-1})}$, and uses $\mytO{\epsilon^{-1}\operatorname{log}(\eta^{-1})}$ extra elementary quantum gates.
\end{itemize}

\section{Numerical simulation}
\label{sec:numerical_results}

In this section, we numerically assess the performance of our GSEE algorithm against an alternative approach with comparable circuit depth and overall runtime. The simulation results show that our algorithm outperforms the alternative, requiring both shallower circuits and shorter runtime.

\subsection{An alternative method based on robust phase estimation}
\label{subsec:alternative_method}
Consider the following GSEE algorithm, which builds directly on prior work. This algorithm utilizes Algorithm \ref{alg:basic_gsee} to produce an initial coarse estimate of \( E_0 \) but incorporates a distinct subroutine to refine this estimate. Specifically, this subroutine combines the ground state preparation method from Ref.~\cite{dong2022ground} with the phase estimation technique from Ref.~\cite{ni2023low} optimized for scenarios with large overlap.

Specifically, we begin by running Algorithm \ref{alg:basic_gsee} to obtain an interval $[l, r]$ of length at most $\Delta/2$ that contains $E_0$ with probability $1-\myTheta{\delta}$. With this information, we then apply the algorithm in Theorem 6 of Ref.~\cite{dong2022ground} to prepare a state $\ket{\phi}$ such that $|\braket{E_0}{\phi}|^2=1-\myTheta{\epsilon/\Delta}$, using $\ket{\psi}$ as the initial state for each circuit in this algorithm. This procedure requires
$$\mytO{\eta^{-1}}$$ quantum circuits, with each circuit evolving $H$ up to time $$\mytO{\Delta^{-1} \mylog{\eta^{-1}\Delta\epsilon^{-1}}}.$$

This state preparation enables us to apply Algorithm 2 in Ref.~\cite{ni2023low}, also known as Robust Phase Estimation (RPE), to estimate \( E_0 \) with accuracy \( \epsilon \) and confidence \( 1 - \myTheta{\delta} \). The RPE algorithm requires that the initial state has a substantial overlap with the ground state \( \ket{E_0} \) — specifically, an overlap of at least \( 4 - 2\sqrt{3} \approx 0.536 \). We fulfill this requirement by using \( \ket{\phi} \) from the previous step as the initial state. The performance of RPE depends on a parameter $\xi$: a smaller $\xi$ results in shallower circuits but longer overall runtime, whereas a larger $\xi$ leads to deeper circuits and shorter overall runtime. Here we set $\xi = \myTheta{\epsilon / \Delta}$ to ensure that the circuit depth remains within the desired regime. Consequently, this procedure requires 
$$ \mytO{\Delta^2 \epsilon^{-2} \mylog{\Delta^{-1}} \mylog{\delta^{-1}}} $$
quantum circuits, each evolving \( H \) up to time
$$ \mytO{\Delta^{-1}}. $$ 
Note that every copy of $\ket{\phi}$ demanded by RPE is prepared by the previous step. Thus, we need to repeat that step 
$$ \mytO{\Delta^2 \epsilon^{-2} \mylog{\Delta^{-1}} \operatorname{log}\lrb{\delta^{-1}}} $$
times to generate all the necessary states.

Taking into account the costs of all components, this alternative GSEE algorithm requires 
$$ \mytO{\eta^{-1}\Delta^2 \epsilon^{-2} \mylog{\Delta^{-1}} \operatorname{log}\lrb{\delta^{-1}}} $$
quantum circuits, where each circuit evolves \( H \) up to time
$$\mytO{\Delta^{-1} \mylog{\eta^{-1}\Delta\epsilon^{-1}}}.$$

For convenience, we will refer to the above algorithm as Algorithm 2-alt, and to its refinement stage (i.e., all steps except the initial stage for obtaining \( [l, r] \)) as Algorithm 3-alt. We will compare the maximal and total evolution times of Algorithm \ref{alg:gsee} (or Algorithm \ref{alg:gsee_ref}) with those of Algorithm 2-alt (or Algorithm 3-alt) in our experiments.

\subsection{Experimental setup and results}
\label{subsec:numerical_results}
In our experiments, we consider the one-dimensional transverse field Ising model (TFIM) defined on \( L \) sites with periodic boundary conditions:
\begin{align}
    H = -\left(\sum_{i=1}^{L-1} Z_i Z_{i+1} + Z_L Z_1\right) - g \sum_{i=1}^L X_i,
\end{align}
where \( X_i \) and \( Z_i \) are Pauli operators acting on the \( i \)-th site, and \( g \) is the coupling coefficient. We normalize this Hamiltonian
so that the eigenvalues lie within \( [-\pi/2, \pi/2] \). Then we estimate the ground-state energy of the rescaled Hamiltonian \( \tilde{H} \coloneqq \frac{\pi H}{2 \norm{H}} \) to within various accuracies. The circuit depth and overall runtime of the tested algorithms are measured in terms of the maximal and total evolution times of \( \tilde{H} \), respectively.

In Figure \ref{fig:compare_gsee_algs}, we present the performance of our algorithm (i.e., Algorithm \ref{alg:gsee} in Section \ref{sec:gsee_alg}) and the alternative algorithm (i.e., Algorithm 2-alt in Section \ref{subsec:alternative_method}) for the TFIM with \( L = 8 \) and \( g = 4 \). We select three different initial states with overlaps of $0.3$, $0.1$, and $0.03$ with the ground state, respectively. In Algorithm \ref{alg:gsee}, we set \( \omega = 0.99 \), \( c = 0.6 \), \( \beta = 0.5 \), and \( \zeta = 0.2 \). For Algorithm 2-alt, we choose \( \xi = \min(10 \epsilon/\Delta, 0.9) \) in the RPE subroutine. We apply the procedure outlined in Appendix \ref{sec:minimize_degree} to determine the minimal-degree trigonometric polynomials necessary for eigenstate filtering within these algorithms.

\begin{figure}[H]
  \centering
  \begin{tabular}{cc}
    \includegraphics[scale=0.45]{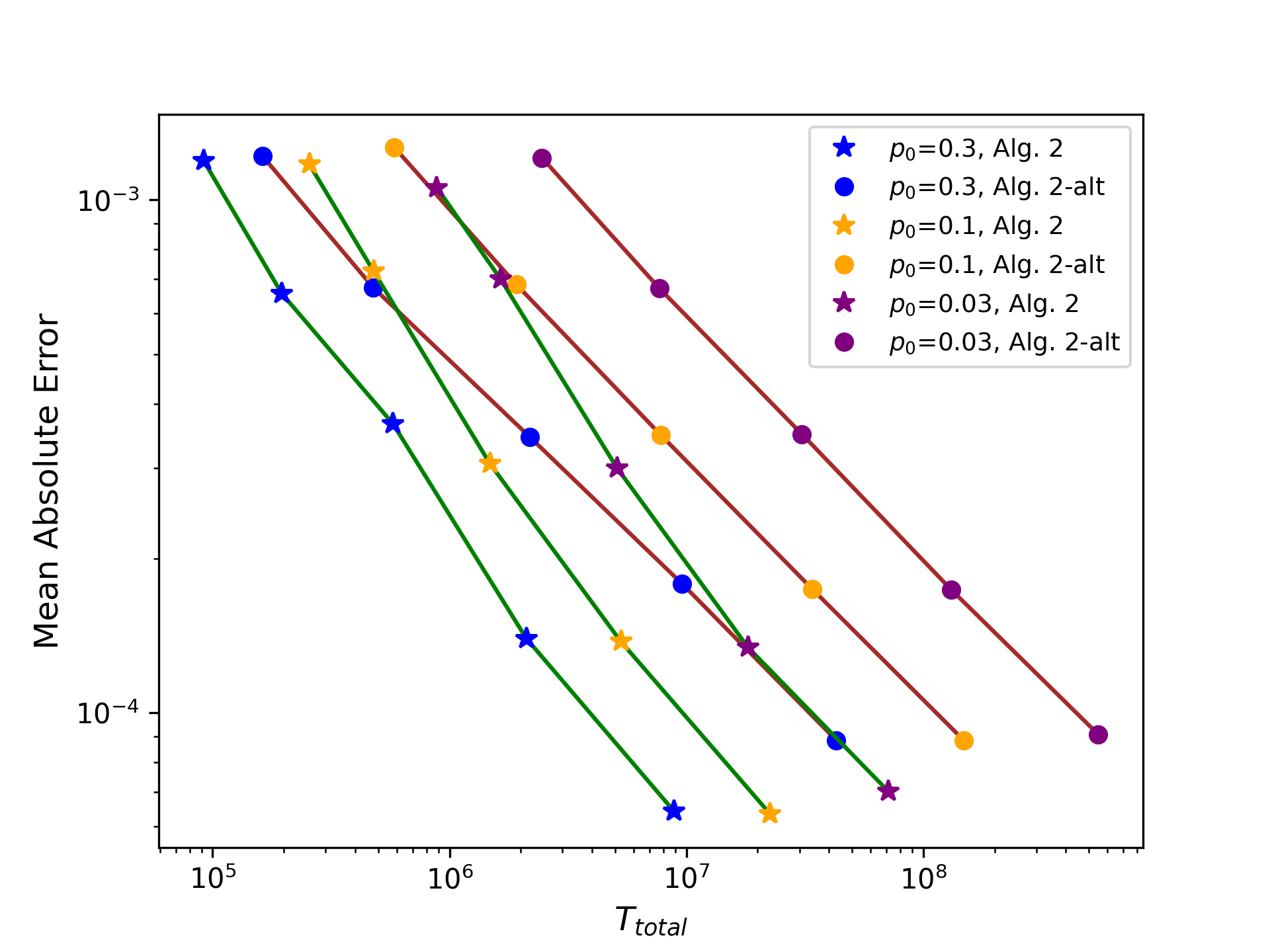} &
    \includegraphics[scale=0.45]{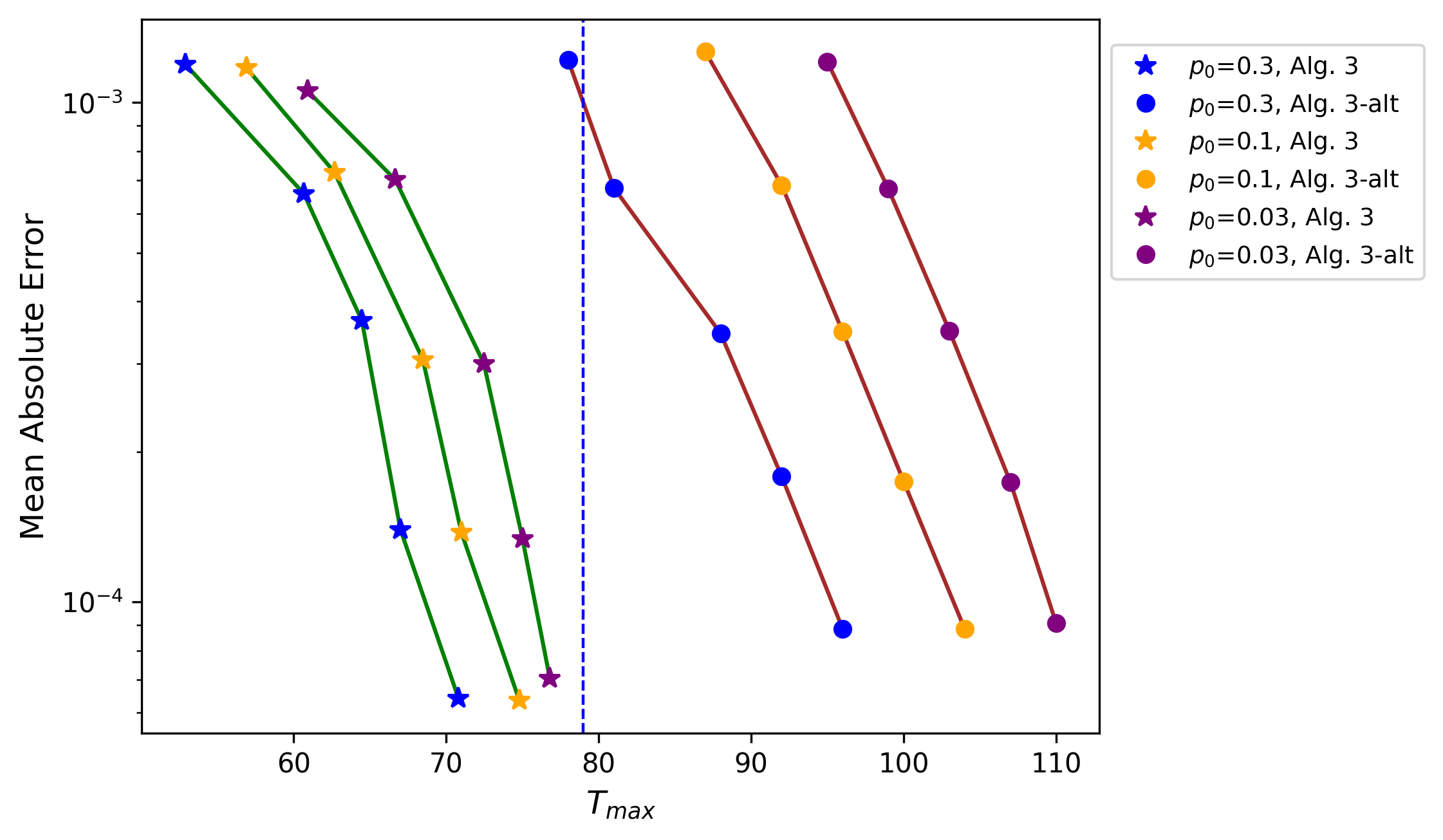}\\    
    (a) & (b)
  \end{tabular}
    \caption{Comparison of our Algorithm \ref{alg:gsee} and the alternative Algorithm 2-alt for the TIFM with $L=8$ and $g=4$. The initial overlap is set to $p_0=0.3$, $0.1$, and $0.03$. (a) Overall runtime $T_{\mathrm{total}}$ of the algorithms. (b) Circuit depth $T_{\mathrm{max}}$ during the refinement stages. Both algorithms share the same initial stage, and its circuit depth is indicated by the dashed vertical line.}
  \label{fig:compare_gsee_algs}
\end{figure}

Algorithms \ref{alg:gsee} and 2-alt share the same initial stage for obtaining a coarse estimate of \( E_0 \); they differ only in their refinement stages, specifically Algorithms \ref{alg:gsee_ref} and 3-alt, respectively. For each setting, we simulate both algorithms 1 million times and compute the mean absolute error of their outputs. As shown in Figure \ref{fig:compare_gsee_algs}, our algorithm achieves the same accuracy with less runtime than the alternative algorithm, despite using shallower circuits. In fact, the refinement stage of our algorithm uses circuits that are shallower than those of the initial stage, whereas the refinement stage of the alternative algorithm requires deeper circuits. We remark that it is possible to reduce the circuit depth of Algorithm 2-alt by choosing smaller \( \xi \), but this would increase its runtime and further widen the performance gap between Algorithms \ref{alg:gsee} and 2-alt.

Figure \ref{fig:compare_gsee_algs} also demonstrates that a smaller $p_0$ leads to a higher runtime, with the runtime scaling linearly in the $1/p_0$. This result is consistent with our theoretical complexity analysis.

We also test the impact of the parameter \( \omega \) on the performance of Algorithm \ref{alg:gsee} as follows. We consider the TFIM with \( L = 10 \) and \( g = 2 \), and fix \( p_0 = 0.1 \), \( c = 0.6 \), \( \beta = 0.5 \), and \( \zeta = 0.2 \), while varying \( \omega = 0.3, 0.6, 0.9 \), respectively. Figure \ref{fig:omega_impact} illustrates the simulation results. As expected, smaller \( \omega \) leads to shallower circuits but longer runtimes, and conversely, larger \( \omega \) leads to deeper circuits but shorter runtimes. Note also that the circuit depth of the refinement stage is often smaller than that of the initial stage unless \( \omega \) is close to 1 and the targe accuracy is high.

\begin{figure}[H]
  \centering
  \begin{tabular}{cc}
    \includegraphics[scale=0.45]{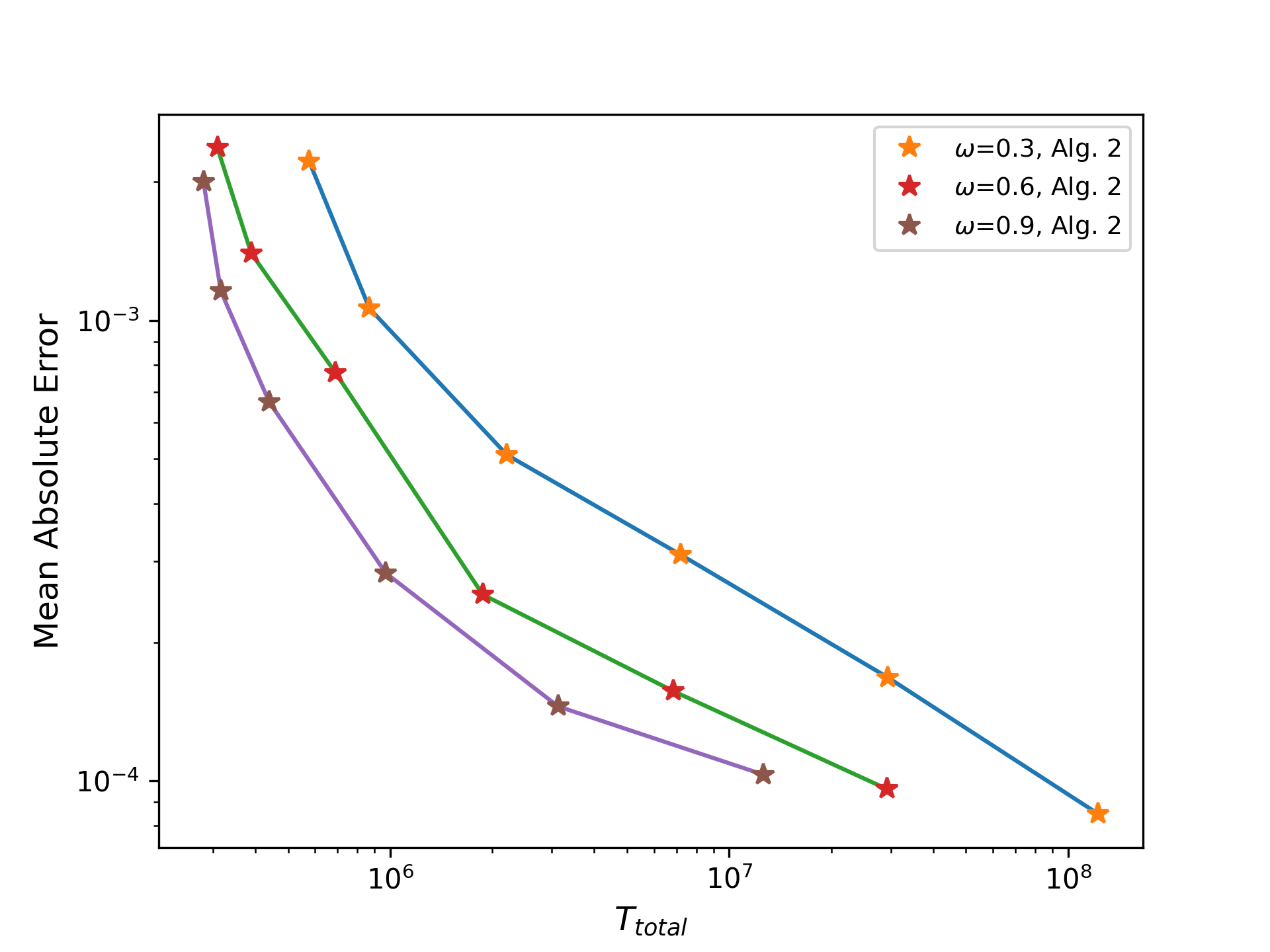} &
    \includegraphics[scale=0.45]{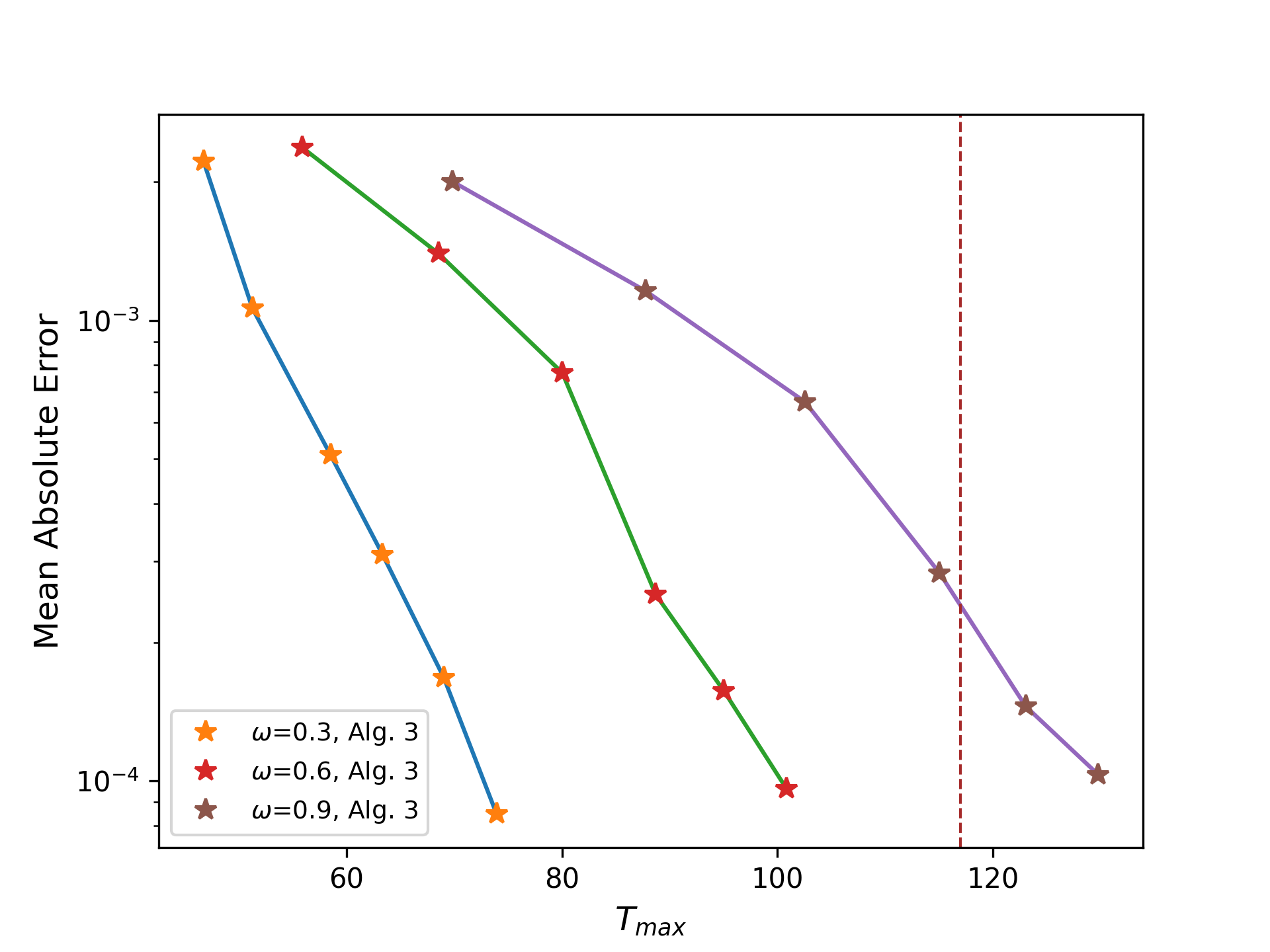}\\    
    (a) & (b)
  \end{tabular}
    \caption{Impact of the parameter $\omega$ on Algorithm \ref{alg:gsee} for the TFIM with $L=10$ and $g=2$. The initial overlap is set to $p_0=0.1$. (a) Overall runtime $T_{\mathrm{total}}$ of the algorithm. (b) Circuit depth $T_{\mathrm{max}}$ during the refinement stage. The circuit depth of the initial stage is indicated by the dashed vertical line.}
  \label{fig:omega_impact}
\end{figure}

\section{Certifying ground-state energy estimates via rejection sampling}\label{sec:certifcation_sec}
\label{sec:certification_result}
In this section, we propose a novel method for performing rejection sampling on quantum computers and develop a quantum algorithm for certifying the accuracy of ground-state energy estimates based on this technique.

\subsection{Rejection sampling from the spectral measure}
\label{subsec:rejection_sampling}
Given the spectral measure $p$ associated with an initial state $\ketbra{\psi}{\psi}$ and a Hamiltonian $H=\sum_i\ketbra{E_i}{E_i}$, i.e.
\begin{align}
p(x)=\sum_ip_i\delta(x-E_i),\quad p_i=|\langle E_i|\psi\rangle|^2,
\end{align}
where \( \delta \) denotes the Dirac delta function, we will show how to sample from a random variable with distribution $p\ast \nu$ on a quantum computer, where $\nu$ is another probability measure and $\ast$ denotes the convolution. It is then not difficult to see that the density of $p\ast \nu$ is given by:
\begin{align}
    (p\ast \nu)(x)=\sum_ip_i\nu(x-E_i).
\end{align}
Sampling from such distributions for various choices of $\nu$ gives us access to information about the spectrum of the Hamiltonian. For instance, if $\nu$ is itself close to a delta distribution at $0$, then sampling from this distribution is close to sampling from $p$ itself. This will, roughly speaking, correspond to our algorithm later, where we will pick $
\nu$ to be a Gaussian with small variance. 

Let us now discuss the sampling routine. We will assume that we can \emph{classically} sample from a random variable $X$ distributed according to a density $\mu$ such that $\textrm{supp}(p\ast \nu)\subset \textrm{supp}(\mu)$. Furthermore, for every $x \in \textrm{supp}(\mu)$, we can implement a $(1, m, 0)$-block-encoding $U_x$ of the operator $\frac{\tilde{\nu}(H-xI)}{M\tilde{\mu}(x)}$, where $\tilde{\nu}$ and $\tilde{\mu}$ satisfy:
\begin{align}\label{equ:requirement_block}
\left|\frac{\tilde{\nu}(H-xI)}{M\tilde{\mu}(x)}\right|^2=\frac{\nu(H-xI)}{|M|^2\mu(x)},
\end{align}
and $M\in \C$ is a constant such that for all $x\in\textrm{supp}(\mu)$:
\begin{align}
\sum_ip_i\frac{\nu(E_i-x)}{|M|^2\mu(x)}\leq 1.
\end{align}
There are many cases where explicit block-encoding circuits have been developed \cite{lee2021even, loaiza2022reducing}.

Let us now explain the quantum algorithm that allows us to sample from $p\ast \nu$ in this setting with an expected number $\mathcal{O}(|M|^2)$ runs of such $U_x$'s.
First, we generate a sample $x$ from the density $\mu$. After that, we implement the block-encoding $U_x$ of $\frac{\tilde{\nu}(H-xI)}{M\tilde{\mu}(x)}$ and apply it to $\ket{0^m}\ket{\psi}$. Then we measure the $m$ ancilla qubits of the resulting state. If the outcome is $0^m$, we accept the sample. Let us see why, conditioned on the acceptance, $x$ will be distributed according to $p\ast \nu$.

Note that:
\begin{align}
U_x\ket{0^m}\ket{\psi}=\ket{0^m}\frac{\tilde{\nu}(H-xI)}{M\tilde{\mu}(x)}\ket{\psi}+\sum_{j \in \zo^m,~j \neq 0^m }\ket{j}\ket{\psi_j'},
\end{align}
where the $\ket{\psi_j'}$'s are unnormalized states. By expanding $\frac{\tilde{\nu}(H-xI)}{M\tilde{\mu}(x)}\ket{\psi}$ in the eigenbasis of $H$, we see that the probability of measuring $0^m$ on the first register is expressed by:
\begin{align}
    \sum_i\left|\frac{\tilde{\nu}(E_i-x)}{M\tilde{\mu}(x)}\right|^2|\langle E_i|\psi\rangle|^2=\sum_ip_i\frac{\nu(E_i-x)}{|M|^2\mu(x)}=\frac{(p\ast \nu)(x)}{|M|^2\mu(x)}.
\end{align}
Thus, the probability of acceptance is nearly the same as if we ran the rejection sampling algorithm for the densities \( p \ast \nu \) and \( \mu \), except that it is now rescaled by a factor of \( 1/|M|^2 \). Furthermore, by standard properties of rejection sampling, we will need \( \mathcal{O}(|M|^2) \) trials in expectation to generate a single sample from the target distribution.

From the discussion above, we obtain our new rejection sampling routine:
\begin{theorem}
    Let  $H=\sum\limits_{i}E_i\ketbra{E_i}{E_i}$ be a Hamiltonian on $n$ qubits and $\ket{\psi}$ an $n$-qubit state. Define the probability measure $p$ on $\mathbb{R}$ as 
    \begin{align}
        p(x)=\sum_ip_i\delta(x-E_i),\quad p_i=|\langle E_i|\psi\rangle|^2,
        \end{align}
    where \( \delta \) denotes the Dirac delta function.
    Furthermore, let $\nu,\mu$ be two probability density functions on $\mathbb{R}$ such that $\textrm{supp}(p\ast \nu)\subset\textrm{supp}(\mu)$, where $\ast$ denotes the convolution. Finally, assume that there exists a constant $M\in \mathbb{C}$ such that, for every $x \in\textrm{supp}(\mu)$, we have access to a $(1,m,0)$-block-encoding  $U_x$ of an operator $h(H-x I)$ which satisfies
    \begin{align}
    |h(H-x I)|^2=\frac{\nu(H-x I)}{|M|^2\mu(x)},\quad \|h(H-x I)\|\leq 1,
    \end{align}
    and that we can generate samples distributed according to $\mu$. Then we can generate a sample distributed according to $p\ast \nu$ from an expected number $|M|^2$ uses of such $U_x$'s.
    \end{theorem}

Let us illustrate the rejection sampling routine using one of the central examples in this work: conditioned normal random variables. Suppose we wish to sample from the truncation of \( p \ast n_\sigma \) to the interval \( [a, b] \), where \( n_{\sigma} \) is the density of a normal random variable with mean \( 0 \) and variance \( \sigma^2 \). One possible strategy to sample from this distribution is to perform rejection sampling with the density \( \mu \) given by the uniform distribution over \( [a, b] \). Then we implement a block-encoding $U_{\sigma, \xi}$ of the operator
\begin{align}
    h_\sigma(H - \xi I) = e^{-\frac{(H - \xi I)^2}{4\sigma^2}},
\end{align}
where \( h_{\sigma}(x) = e^{-x^2/(4\sigma^2)} \). This function satisfies the requirement in Eq.~\eqref{equ:requirement_block} with \( |M|^2 \) given by
\begin{align}\label{equ:constant_rejection}
    |M|^2 = \frac{b - a}{\sqrt{2\pi} \sigma \int_{a}^{b}(p \ast n_\sigma)(x)  dx}.
\end{align}
With this choice of \( |M|^2 \), the ratio of the density \( p \ast n_\sigma|_{[a,b]} \) to the uniform distribution on \( [a, b] \) corresponds to the acceptance probability. Moreover, \( \|h_{\sigma}(H - \xi I)\| \) is clearly bounded above by \( 1 \) for all \( \xi \). Note that we do not need to know the constant in Eq.~\eqref{equ:constant_rejection} to execute the algorithm, only \( h_\sigma \). However, \( |M|^2 \) will determine the number of repetitions required.

In Appendix \ref{sec:implement_gaussian}, we prove that QET-U can be used to implement a block-encoding of $h_{\sigma}(H-\xi I)$  with a quantum circuit of depth \( \tilde{\mathcal{O}}(\sigma^{-1}) \). Consequently, we can draw samples from the truncated mixture of Gaussians $p \ast n_{\sigma}$ over an interval $[a,b]$ by employing multiple circuits of this depth.

We believe that the rejection sampling subroutine developed in this work is of independent interest and will find applications beyond ground-state energy certification. Indeed, to the best of our knowledge, it is the first quantum algorithm for sampling from distributions associated with the spectral measure with \emph{a constant number} of ancilla qubits.

\subsection{ground-state energy certification algorithm}
\label{subsec:gse_cert_alg}
Next, we propose a method for certifying the accuracy of the ground-state energy estimate produced by Algorithm~\ref{alg:gsee} or other GSEE algorithms. This approach is especially useful when there is uncertainty about the correctness of \( \Delta \), the lower bound on the energy gap. Furthermore, if the input state has additional structure (i.e., a sufficiently small weight on low-energy excited states), our algorithm can provide a correct estimate for evolution times that are arbitrarily shorter than the inverse of the spectral gap. To illustrate, consider the case where the input state is the ground state itself. In this situation, any constant evolution time suffices to obtain a reliable estimate. In both scenarios, it can be advantageous to assume that the algorithm parameters are correct, run Algorithm~\ref{alg:gsee}, and then certify the result. This way, we do not have to rely solely on our estimate of the spectral gap and can benefit from better-prepared input states.

We will describe a procedure that achieves exactly this. Specifically, given a lower bound \( \eta \) on the overlap of the input state with the ground state, access to samples from the distribution \( p \ast n_\sigma|_{[a,b]} \), an error tolerance \( \epsilon > 0 \), and an initial estimate \( \hat{E}_0 \) of the ground-state energy \( E_0 \) accurate within \( \sigma \), our procedure generates a new estimate \( \hat{E}'_0 \) of \( E_0 \) and certifies whether \( |\hat{E}'_0 - E_0| \leq \epsilon \), with a failure probability of at most \( \delta \) for incorrect certification. 

It is important to emphasize that our test \emph{does not certify} a spectral gap estimate; it certifies only the ground-state energy estimate. However, we will show that if $\sigma$ is proportional to the spectral gap, the test will accept the estimate \( \hat{E}'_0 \) with high probability. It is also worth noting that this test incurs a higher sample complexity than the estimation procedure itself, specifically \( \tilde{\mathcal{O}}(\epsilon^{-4} \eta^{-3}) \).

We provide an intuitive explanation of the verification process in Figures~\ref{fig:scheme1}, \ref{fig:scheme2}, \ref{fig:scheme3} and \ref{fig:scheme4}, and explain the protocol in detail in Appendix~\ref{app:certification}. The verification relies primarily on two observations. The first, formalized in Lemma~\ref{lem:contribution_small}, is that if the density of \( p \ast n_\sigma \) has a well-defined peak near our ground-state energy estimate $\hat{E}_0$, then \( p \ast n_\sigma \), conditioned on an interval of size \( \sim \sigma \) around \( \hat{E}_0 \), will approximately exhibit the moments of a convex combination of Gaussians. Furthermore, we demonstrate that the variance of the random variable conditioned on this interval will exceed a certain threshold if the estimate is incorrect. Finally, we show that if $\sigma$ is proportional to the spectral gap, the test will accept the estimate with high probability. We summarize the certification algorithm below:

\begin{theorem}
Under the same conditions as Theorem \ref{thm:main_in_he}, there exists an algorithm that, given arbitrary $\sigma>0$ and $\hat{E}_0 \in [E_0-\sigma, E_0+\sigma]$ (but no bound on the spectral gap $\Delta_{\rm true}$), uses
\begin{align}
\tilde{\mathcal{O}}(\eta^{-3}\epsilon^{-4}\mylog{\delta^{-1}})
\end{align}
block-encodings of Gaussian functions of $H$, as provided in Corollary \ref{cor:qetu_circuit_shifted_gaussian}, in expectation, and generates a new estimate $\hat{E}_0'$ of $E_0$. If $|\hat{E}_0'-E_0|\geq \epsilon$, it rejects $\hat{E}_0'$ with probability at least $1-\delta$. If it accepts, then $|\hat{E}_0'-E_0|\leq \epsilon$, with probability of incorrect acceptance at most $\delta$. Furthermore, there exists some
\begin{align}\label{equ:minimal_time}
\sigma_0=\Omega\left(\frac{\Delta_{\rm{true}}}{\mylog{\epsilon^{-1}\eta^{-1}}}\right)
\end{align}
such that if $\sigma\le \sigma_0$, then this algorithm outputs an estimate $\hat{E}_0'$ that it accepts with probability $1-\delta$ and is such that $|\hat{E}_0'-E_0|\leq \epsilon$.
\end{theorem}
We defer all proofs and detailed descriptions of the algorithm to Appendix~\ref{app:certification}. Note that the minimum evolution time required for the test to accept, as implied by Eq.~\eqref{equ:minimal_time}, is of the same order as the time at which our algorithm is guaranteed to succeed. However, the sample complexity of the test is higher than that of the algorithm at this depth. As a result, the test is advantageous only when we cannot fully rely on our estimate of the spectral gap or when we have reason to believe that the input state allows for a good estimate with shorter circuit depths.

Our algorithm is based on empirically estimating the variance of a random variable with a certain precision. As is typical in such protocols, the empirical estimate has a tolerated error that depends on the number of samples, up to a specified failure probability. Notably, our protocol consistently achieves accurate rejections when variance estimates fall within the tolerated error range. Therefore, the parameter \( \delta \) quantifies the probability that the estimates deviate beyond the acceptable threshold for verification, meaning that we cannot guarantee performance when inaccurate estimates are provided.

Additionally, the test can be combined with a binary search procedure to determine \( \sigma_0 \), the largest \( \sigma \) for which the test accepts. This approach could yield reliable ground-state energy estimates with depths potentially much smaller than \( \Delta_{\rm true}^{-1} \). However, the sample complexity of this algorithm would increase to \( \tilde{\mathcal{O}}(\epsilon^{-4} \eta^{-3}) \), which may still be advantageous if we aim to optimize the maximum circuit depth.

  \begin{center}
    \begin{minipage}{0.35\columnwidth}
    \centering
    \begin{tikzpicture}[scale=1.5]
        \draw[->] (-2,0) -- (1,0) node[right] {$x$};
        \draw[->] (0,-0.5) -- (0,1.5) node[above] {$y$};
        \draw[domain=-2:1,samples=200,smooth,variable=\x,blue] plot ({\x},{exp(-\x*\x/0.05)+0.3*exp(-(\x+1)*(\x+1)/0.05)});
        \draw[<-] (-1.05,0.36) -- (-1.7,0.7) node[left] {$\hat{E}_0$};
        \node[above right] at (0.2,1) {Peaked};
        \node[above right] at (0.2,0.5) {\color{green}{Check width}};
      
      \end{tikzpicture}
    \captionof{figure}{Example of density of mixture of Gaussians. In this case, the density quickly goes to $0$ around our estimate $\hat{E}_0$. We then continue to check the variance conditioned around $\hat{E}_0$ to see if we accept the estimate.}
    \label{fig:scheme1}
    \end{minipage}
\hspace{0.05\textwidth}
\begin{minipage}{0.35\columnwidth}
    \centering
    \begin{tikzpicture}[scale=1.5]
        \draw[->] (-2,0) -- (1,0) node[right] {$x$};
        \draw[->] (0,-0.5) -- (0,1.5) node[above] {$y$};
        \draw[domain=-2:1,samples=200,smooth,variable=\x,blue] plot ({\x},{exp(-(\x+0.3)*(\x+0.3)/0.05)+0.3*exp(-(\x+1)*(\x+1)/0.05)});
        \draw[<-] (-1.05,0.36) -- (-1.7,0.7) node[left] {$\hat{E}_0$};
        \node[above right] at (0.2,1) {Not peaked};
        \node[above right] at (0.2,0.5) {\color{red}{Reject}};
      
      \end{tikzpicture}
      \captionof{figure}{Example of density of mixture of Gaussians. In this case, the density quickly \emph{does not} go to $0$ quickly around $\hat{E}_0$. We reject the estimate.}
      \label{fig:scheme2}
    \end{minipage}
    \hspace{0.05\textwidth}
\begin{minipage}{0.35\columnwidth}
    \centering
    \begin{tikzpicture}[scale=1.5]
        \draw[->] (-2,0) -- (1,0) node[right] {$x$};
        \draw[->] (0,-0.5) -- (0,1.5) node[above] {$y$};
        \draw[domain=-2:1,samples=200,smooth,variable=\x,blue] plot ({\x},{exp(-\x*\x/0.05)+0.3*exp(-(\x+1)*(\x+1)/0.05)});
        \draw[<-] (-1.05,0.36) -- (-1.7,0.7) node[left] {$\hat{E}_0$};
        \draw[green] (-1,0.3) circle [radius=0.5];
        \node[above right] at (0.2,1) {\color{black}{Condition}};
        \node[above right] at (0.2,0.7) {\color{black}{around peak}};
      
      \end{tikzpicture}
    \captionof{figure}{As the distribution is peaked around $\hat{E}_0$, we condition around it in the next step of the verification.}
    \label{fig:scheme3}    
    \end{minipage}
    \hspace{0.05\textwidth}
\begin{minipage}[t]{0.35\columnwidth}
    \centering
    \begin{tikzpicture}[scale=1.5]
        \draw[->] (-2,0) -- (-0.4,0) node[right] {$x$};
        \draw[domain=-2:-0.5,samples=200,smooth,variable=\x,blue] plot ({\x},{exp(-\x*\x/0.05)+1*exp(-(\x+1)*(\x+1)/0.05)});
        \draw[<-] (-1.05,1.1) -- (-1.7,1.7) node[left] {$\hat{E}_0$};
        %\draw[green] (-1,0.3) circle [radius=0.5];
        \node[above right] at (-0.4,1) {\color{black}{Estimate}};
        \node[above right] at (-0.4,0.7) {\color{black}{variance}};
        \draw[<->] (-0.8,0.2) -- (-1.2,0.2) node[left] {$\sigma^2$};
      
      \end{tikzpicture}
    \captionof{figure}{We estimate the variance of the random variable conditioned around $\hat{E}_0$. If it is close to $\sigma^2$, we accept it. If it differs from $\sigma^2$ by $\sim\eta\epsilon^2$, we reject it.}
    \label{fig:scheme4}
    \end{minipage}
    \end{center}

\section{Conclusion and outlook}
\label{sec:conclusion}
To summarize, we have introduced a quantum algorithm for estimating the ground-state energy of a Hamiltonian using fewer operations per circuit than previous methods, making it more suitable for early fault-tolerant quantum computation. The circuit depth required by our algorithm depends on the spectral gap of the Hamiltonian rather than the target accuracy, and its overall runtime is shorter than that of previous algorithms with similar circuit depths. 

Furthermore, this algorithm can be adjusted to trade circuit depth for runtime, and its runtime can be further reduced by parallelizing sample collection across multiple quantum devices (i.e., trading time resources for space resources). These features make our algorithm a strong candidate for achieving quantum advantage in industrially-relevant problems on early fault-tolerant quantum computers \cite{wang2019accelerated, wang2021minimizing, lin2022heisenberg, dong2022ground, zhang2022computing, wang2022state, wang2023quantum, wang2023qubit}.

Additionally, we have proposed a novel method for performing rejection sampling on quantum computers and devised the first approach to certifying ground-state energy estimates in the regime of circuit depth \( \tilde{\mathcal{O}}(\Delta_{\rm true}^{-1}) \). While the current method has a high sample complexity, it enables certifiably correct ground-state energy estimates from arbitrarily small circuit depths, depending on the properties of the initial state.

Finally, we would like to highlight several research directions that merit further exploration:

\begin{itemize}

\item In this work, we have focused on GSEE in the regime of circuit depth \( \ge 1/\Delta_{\rm true} \). However, it remains unclear how challenging this problem is when the circuit depth is \( \ll 1/\Delta_{\rm true} \). Is it possible to devise an algorithm with reasonable runtime (e.g., \( \operatorname{poly}(1/\epsilon, 1/p_0) \)) in this setting?

\item While we have made progress in developing novel low-depth quantum algorithms for GSEE, an important future direction is to establish lower bounds on the runtimes of such algorithms, particularly in the case of \( \tilde{\mathcal{O}}(\Delta^{-1}) \) circuit depth. This would deepen our understanding of the strengths and limitations of early fault-tolerant quantum computers for simulating quantum systems. 

\item Estimating ground state properties beyond their energies is essential in many applications. However, this problem has not yet been extensively studied in the low-depth regime. An exception is Ref.~\cite{zhang2022computing}, which develops a low-depth algorithm for estimating the expectation value of an observable \( O \) with respect to the ground state of a Hamiltonian \( H \). This problem is known as \emph{ground state property estimation} (GSPE). To estimate \( \bra{E_0} O \ket{E_0} \) within additive error \( \epsilon \) with high probability, the algorithm of Ref.~\cite{zhang2022computing} requires circuit depth \( \tilde{\mathcal{O}}(\Delta^{-1} \operatorname{polylog}(\epsilon^{-1} p_0^{-1})) \) and overall runtime \( \tilde{\mathcal{O}}(\Delta^{-1} \epsilon^{-2} p_0^{-2}) \). Notably, this runtime is worse than our runtime for ground-state energy estimation. Since GSEE can be viewed as a special case of GSPE when \( O = H \), could the techniques developed in this paper improve the efficiency of GSPE?

\item Can we design more efficient methods for certifying the correctness of ground-state energy estimates? Our current certification method has a quadratically higher sample complexity than our GSEE algorithm. If we could develop a certification method with the same sample complexity as the GSEE algorithm, we could employ a binary search to find the minimal evolution time required for reliable ground-state energy estimates without additional assumptions.

\item Finally, we expect that our rejection sampling method could find applications beyond ground-state energy certification. In many contexts, one needs to generate samples from a continuous distribution \( X \) with probability density \( \nu(x) \). The traditional approach is to prepare a quantum state approximating \( \int_\Omega \sqrt{\nu(x)} \ket{x}  dx \) and then measure it. However, this approach has two drawbacks. First, creating this state may be challenging. Second, it requires discretizing the target distribution, where the number of grid points (i.e., the system's dimension) depends on the desired accuracy. Our approach avoids this discretization by encoding the sample value directly into the circuit parameter. Moreover, our circuit might be easier to implement than the direct approach and may require only \( \mathcal{O}(1) \) qubits, regardless of the target accuracy. Provided that rejection sampling is efficient (i.e., few raw samples are discarded), our method could accelerate this central component of many algorithms, potentially enhancing its suitability for near-term implementations.

\end{itemize}

We have contributed to the growing body of research on quantum algorithms for early fault-tolerant quantum computers. We hope that these results will inspire further exploration toward discovering the quantum algorithms that will achieve quantum advantage first.

\section*{Acknowledgements}
We thank the anonymous referees for their valuable comments on earlier versions of this paper and for suggesting the alternative approach discussed in Section \ref{subsec:alternative_method}.

\appendix 

\section{Minimizing degree in trigonometric approximations}
\label{sec:minimize_degree}
To utilize QET-U to implement the operations $V_{a,b}$ in Algorithm \ref{alg:basic_gsee} and $V$ in Algorithm \ref{alg:gsee_ref}, we need to find  trigonometric polynomials of the form $g(x) \defeq \sum_{j=0}^d \alpha_j \mycos{jx}$ that satisfy certain conditions. Specifically, these conditions require that
\begin{itemize}
    \item \( |g(x)| \le 1 \) for \( x \in [0, \pi] \);
    \item \( |g(x)| \le \epsilon \) for \( x \in [l_0, r_0] \), where $0 \le l_0<r_0 \le \pi$ and $\epsilon \in (0, 1)$ are given;  
    \item (for $V_{a,b}$ only) \( |g(x) - 1| \le \epsilon \) for \( x \in [l_1, r_1] \), where $0 \le l_1<r_1 \le \pi$ and $\epsilon \in (0, 1)$ are given;
    \item (for $V$ only) $g(x) \ge c'$ for $x \in [l', r']$, where $0 \le l'<r' \le \pi$ and $c' \in (0,1)$ are given.
\end{itemize}
Although Lemma \ref{lem:fourier_approx_threshold} provides explicit constructions for such polynomials with degree $\mytO{\frac{\mylog{1/\epsilon}}{\Delta}}$, this degree might be unnecessarily large. Instead, we can find the minimal-degree trigonometric polynomial that satisfies the above conditions by numerically solving a sequence of linear programming problems. This approach helps reduce the cost of implementing the operations $V$ and $V_{a,b}$. 

Precisely, for any given $d$, we can numerically test whether there exists a trigonometric polynomial of degree at most $d$ that fulfill the above conditions. First, we select a dense set of grid points ${x_1, x_2, \dots, x_N} \subset [0, \pi]$ \footnote{In our experiments, we choose a grid of $3000$  points in $[0,\pi]$ and find that this is sufficient for our purpose. This choice ensures that the trigonometric polynomial found using this grid meets the conditions not only at the grid points but also at non-grid points within the interval.}. Next, we define a linear programming problem as follows. The variables are $\alpha_0, \alpha_1, \dots, \alpha_d$ and $\delta$. We then introduce the following linear constraints based on the conditions above:
\begin{itemize}
    \item $-1 \le \sum_{j=0}^d \alpha_j \mycos{j x_i} \le 1$, for each $x_i$;
    \item $-\delta \le \sum_{j=0}^d \alpha_j \mycos{j x_i} \le \delta$, for each $x_i \in [l_0, r_0]$;
    \item (for $V_{a,b}$ only)
$\sum_{j=0}^d \alpha_j \mycos{j x_i} \ge 1-\delta$, for each $x_i \in [l_1, r_1]$;  
    \item (for $V$ only)
$\sum_{j=0}^d \alpha_j \mycos{j x_i} \ge c'$, for each $x_i \in [l', r']$.
\end{itemize}
The objective function is to minimize $\delta$. This problem can be efficiently solved using software packages like CVXPY. Let $\delta^*$ denote the optimal value. Then there exists a trigonometric polynomial of degree at most $d$ that meets the desired conditions if and only if $\delta^* \leq \epsilon$. Finally, we perform a binary search to find the smallest degree $d$ for which such a polynomial exists. We note that a similar approach was proposed in Ref.~\cite{dong2022ground}.

\section{Implementing Gaussian functions of the Hamiltonian}
\label{sec:implement_gaussian}
In this appendix, we demonstrate how to use QET-U to implement a block-encoding of a Gaussian function of the Hamiltonian \( H \). Specifically, given \( \sigma, \epsilon > 0 \) and \( \xi \in \mathbb{R} \), our goal is to construct an \( (n + m) \)-qubit unitary operation \( V_{\sigma, \xi} \) such that
\begin{align}
    \norm{\lrb{\bra{0^m} \otimes I} V_{\sigma,\xi} \lrb{\ket{0^m} \otimes I} - e^{-\frac{(H-\xi I)^2}{2\sigma^2}}} \le \epsilon, 
\end{align}
by using controlled time evolution of \( H \) along with elementary quantum gates.

We will use QET-U to accomplish this goal. To do so, we need to find a low-degree trigonometric approximation of Gaussian functions. Specifically, we claim that:

\begin{lemma}
For every $\sigma, \epsilon \in (0, 1)$, there exist an efficiently-computable even real polynomial $G_{\sigma;\epsilon}(x)$ of degree $\mathcal{O}(\sigma^{-1} T \sqrt{\mylog{1/\epsilon}})$, where $T=\max(2\pi, \Theta(\sigma \sqrt{\mylog{1/\epsilon}}))$, such that 
\begin{itemize}
    \item $|G_{\sigma;\epsilon}(\cos(\pi x/T))| \le 1$, for all $x \in \R$;
    \item $|G_{\sigma;\epsilon}(\cos(\pi x/T)) - e^{-x^2/(2\sigma^2)}| \le \epsilon$, for all $x \in [-\pi, \pi]$.
\end{itemize}
\label{lem:trigonometric_approx_gaussian_func}
\end{lemma}
The proof of this lemma will be provided later. Combining Lemma~\ref{lem:qetu} and Lemma~\ref{lem:trigonometric_approx_gaussian_func} yields:
\begin{corollary}
Suppose $H$ is a Hamiltonian with $\norm{H} \le \pi$. Then for every $\sigma, \epsilon \in (0, 1)$, we can implement a $(1, 1, \epsilon)$-block-encoding of $e^{-H^2/(2\sigma^2)}$ by using controlled evolutions of $H$  
for a total time of $\myO{\sigma^{-1} \sqrt{\mylog{1/\epsilon}}}$, 
along with $\mathcal{O}(\sigma^{-1} T \sqrt{\mylog{1/\epsilon}})$ primitive quantum gates, where $T=\max(2\pi, \Theta(\sigma \sqrt{\mylog{1/\epsilon}}))$.
\label{cor:qetu_circuit_gaussian}
\end{corollary}

Note that for any $\xi, t \in \R$, we can implement controlled-$e^{\i(H-\xi I)t}$ by using one controlled-$e^{\i Ht}$ and one $Z$-rotation gate on the control qubit. Combining this fact and Corollary \ref{cor:qetu_circuit_gaussian}, we conclude that:

\begin{corollary}
Suppose $H$ is a Hamiltonian with $\norm{H} \le \pi/2$. Then for every $\sigma, \epsilon \in (0, 1)$, $\xi \in [-\pi/2, \pi/2]$, we can implement a $(1, 1, \epsilon)$-block-encoding of $e^{-(H-\xi I)^2/(2\sigma^2)}$ by using controlled evolutions of $H$  
for a total time of $\myO{\sigma^{-1} \sqrt{\mylog{1/\epsilon}}}$, 
along with $\mathcal{O}(\sigma^{-1} T \sqrt{\mylog{1/\epsilon}})$ primitive quantum gates, where $T=\max(2\pi, \Theta(\sigma \sqrt{\mylog{1/\epsilon}}))$.
 \label{cor:qetu_circuit_shifted_gaussian}
\end{corollary}
This allows us to implement the block-encoding of $h_{\sigma}(H-\xi I)$ in Section \ref{subsec:rejection_sampling}, which is a key component of our ground-state energy certification algorithm.

\begin{proof}[Proof of Lemma \ref{lem:trigonometric_approx_gaussian_func}]
    We claim that for every $\sigma, \epsilon \in (0, 1)$, there exist $T=\max(2\pi, \Theta(\sigma \sqrt{\mylog{1/\epsilon}}))$, $N=\mathcal{O}(\sigma^{-1} T \sqrt{\mylog{1/\epsilon}})$ and
    $a_0, a_1, \dots, a_N \in \R^+$ such that $\sum_{j=0}^N a_j \le 1$ and  
   \begin{align}
   \abs{\sum_{j=0}^N a_j \cos(2\pi j x/T) - e^{-x^2/(2 \sigma^2)}} \le \epsilon,~&~\forall x \in [-\pi, \pi].       
   \end{align}

If this claim holds, then since $\cos(2\pi j x/T) = {\cal T}_{2j}(\cos(\pi x /T))$, where ${\cal T}_{2j}(x)$ is the $2j$-th Chebyshev polynomial of the first kind, $\sum_{j=0}^N a_j \cos(2 \pi jx/T)$ can be written as $G_{\sigma; \epsilon}(\cos(\pi x/T))$ for some even real polynomial $G_{\sigma; \epsilon}(x)$ of degree $2N={\mathcal O}(\sigma^{-1}T\sqrt{\mylog{1/\epsilon}})$. So we get 
\begin{align}
\abs{G_{\sigma; \epsilon'}(\cos(\pi x/T)) - e^{-x^2/(2\sigma^2)}} \le \epsilon, ~&~\forall x \in [-\pi, \pi].    
\end{align}
Furthermore, since $a_0, a_1, \dots, a_N \ge 0$, we have 
\begin{align}
G_{\sigma; \epsilon}(\cos(\pi x/T))=\sum_{j=0}^N a_j \cos(2\pi j x/T) \le \sum_{j=0}^N a_j \le 1,~&~\forall x \in \R.
\end{align}
Therefore $G_{\sigma; \epsilon}(x)$ satisfies all the requirements of Lemma \ref{lem:trigonometric_approx_gaussian_func}.

The proof of the above claim is inspired by the Nyquist–Shannon sampling theorem. Specifically, let $f(x) \defeq \sum_{k=-\infty}^\infty \delta(x-k T)$ be the Dirac comb with period $T$. It has Fourier transform $\hat{f}(\xi) = \frac{1}{T}\sum_{n=-\infty}^\infty \delta(\xi - n/T)$. Meanwhile, let $g(x) \defeq e^{-x^2/(2\sigma^2)}$ be Gaussian with mean $0$ and variance $\sigma^2$. It has Fourier transform $\hat{g}(\xi)=\sqrt{2\pi}\sigma e^{-2\pi^2 \sigma^2 \xi^2}$. Convolving $f$ and $g$ yields
\begin{align}
(f * g)(x) = \sum_{k=-\infty}^\infty g(x-kT).
\label{eq:conv1}
\end{align}
On the other hand, the Fourier transform of $f*g$ is 
\begin{align}
\widehat{f*g}(\xi) = \hat{f}(\xi) \hat{g}(\xi)=  \frac{1}{T}\sum_{n=-\infty}^\infty  \hat{g}(n/T) \delta(\xi - n/T).
\end{align}
Applying inverse Fourier transform to both sides of this equation leads to
\begin{align}
(f * g)(x) = \frac{1}{T}\sum_{n=-\infty}^\infty  \hat{g}(n/T) e^{\i 2\pi x n/T}. 
\label{eq:conv2}
\end{align}
Comparison of Eqs.~\eqref{eq:conv1} and \eqref{eq:conv2} indicates that
\begin{align}
    \sum_{k=-\infty}^\infty g(x-kT) =  \frac{1}{T}\sum_{n=-\infty}^\infty  \hat{g}(n/T) e^{\i 2\pi x n/T}. 
    \label{eq:conv3}
\end{align}
Now we claim that for some $T=\max(2\pi, \Theta(\sigma \sqrt{\mylog{1/\epsilon}}))$ and $N={\mathcal O}(\sigma^{-1} T \sqrt{\mylog{1/\epsilon}})$, 
\begin{align}
     \sum_{k=-\infty}^{-1} g(x-kT) + \sum_{k=1}^\infty g(x-kT)  \le \frac{\epsilon}{6}, ~&~\forall x \in [-\pi, \pi],
\label{eq:trunc_err1}
\end{align}
and
\begin{align}
    \abs{ \frac{1}{T}\sum_{n=-\infty}^{-N-1}  \hat{g}(n/T) e^{\i 2\pi x n/T}
+    \frac{1}{T}\sum_{n=N+1}^\infty  \hat{g}(n/T) e^{\i 2\pi x n/T}}
\le  \frac{\epsilon}{6},~&~\forall x \in \R.
\label{eq:trunc_err2}
\end{align}
If these claims are true, then combining them and Eq. \eqref{eq:conv3} yields
\begin{align}
\abs{ g(x) - \frac{1}{T}\sum_{n=-N}^N  \hat{g}(n/T) e^{\i 2\pi x n/T}}
=\abs{ g(x) - \hat{g}(0) - \frac{2}{T}\sum_{n=1}^N  \hat{g}(n/T) \cos(2\pi x n/T)} \le \frac{\epsilon}{3}, ~&~\forall x \in [-\pi, \pi].
\end{align}
This also implies that 
\begin{align}
\hat{g}(0) + \frac{2}{T}\sum_{n=1}^N  \hat{g}(n/T) \le 1+\frac{\epsilon}{3}.    
\end{align}
Now let $a_0=(1-\epsilon/3)\hat{g}(0)$ and $a_j=2T^{-1}(1-\epsilon/3)\hat{g}(j/T)$ for $j=1,2,\dots,N$. 
Then we have $a_0, a_1, \dots, a_N >0$ and
\begin{align}
    \sum_{j=0}^N a_j = (1-\epsilon/3) \lrb{\hat{g}(0) + \frac{2}{T}\sum_{n=1}^N  \hat{g}(n/T)} \le (1-\epsilon/3)(1+\epsilon/3) < 1,
\end{align}
and for any $x \in [-\pi, \pi]$,
\begin{align}
       \abs{\sum_{j=0}^N a_j \cos(2\pi j x/T) - g(x)}
       &\le
       \abs{\sum_{j=0}^N a_j \cos(2\pi j x/T) - \hat{g}(0) - \frac{2}{T}\sum_{n=1}^N  \hat{g}(n/T) \cos(2\pi x n/T)} \nonumber\\
       &\quad + \abs{\hat{g}(0) + \frac{2}{T}\sum_{n=1}^N  \hat{g}(n/T) \cos(2\pi x n/T) - g(x)}\\
        &\le \frac{\epsilon}{3} \cdot \abs{ \hat{g}(0) + \frac{2}{T}\sum_{n=1}^N  \hat{g}(n/T) \cos(2\pi x n/T)} + \frac{\epsilon}{3}\\
       &\le \frac{\epsilon}{3} \lrb{\hat{g}(0) + \frac{2}{T}\sum_{n=1}^N  \hat{g}(n/T) } + \frac{\epsilon}{3} \\
       &\le \frac{\epsilon}{3} \lrb{1+\frac{\epsilon}{3}} + \frac{\epsilon}{3}\\
       &\le \epsilon,
\end{align}
as desired.

It remains to prove Eqs.~\eqref{eq:trunc_err1} and \eqref{eq:trunc_err2}. To prove the former, we impose the constraint $T \ge 2\pi$. This means that $|x|\le T/2$ for all $x \in [-\pi, \pi]$. Then it follows that
\begin{align}
    \sum_{k=-\infty}^{-1} g(x-kT) + \sum_{k=1}^\infty g(x-kT)  
   & \le 2 \sum_{l=0}^{\infty} g((l+1/2)T)\\
   & = 2 \sum_{l=0}^{\infty} \myexp{-(l+1/2)^2T^2/(2 \sigma^2)} \\
   & = 2\myexp{-T^2/(8 \sigma^2)} \sum_{l=0}^{\infty} \myexp{-(l^2+l)T^2/(2 \sigma^2)} \\
   & \le 2\myexp{-T^2/(8 \sigma^2)} \sum_{l=0}^{\infty} \myexp{-l T^2/\sigma^2} \\
   & = \frac{2\myexp{-T^2/(8 \sigma^2)} }{1 - \myexp{-T^2/\sigma^2} }.
   \label{eq:trunc_err1_bound}
\end{align}
By picking some $T=\Theta(\sigma \sqrt{\mylog{1/\epsilon}})$, we can ensure that the RHS of Eq.~\eqref{eq:trunc_err1_bound} is at most ${\epsilon}/{6}$. Overall, Eqs.~\eqref{eq:trunc_err1} holds as long as we pick some $T=\max(2\pi, \Theta(\sigma \sqrt{\mylog{1/\epsilon}}))$. 

To prove Eq. \eqref{eq:trunc_err2}, we note that
\begin{align}
    \abs{ \frac{1}{T}\sum_{n=-\infty}^{-N-1}  \hat{g}(n/T) e^{\i 2\pi x n/T}
+    \frac{1}{T}\sum_{n=N+1}^\infty  \hat{g}(n/T) e^{\i 2\pi x n/T}}
&\le      \frac{1}{T}\sum_{n=-\infty}^{-N-1}  \hat{g}(n/T)
+    \frac{1}{T}\sum_{n=N+1}^\infty  \hat{g}(n/T) \\
&= \frac{2}{T}\sum_{n=N+1}^\infty  \hat{g}(n/T) \\
&= \frac{2\sqrt{2\pi} \sigma}{T} \sum_{n=N+1}^\infty \myexp{-2 \pi^2 \sigma^2 n^2/T^2} \\
&\le \frac{2\sqrt{2\pi} \sigma}{T} \int_{N}^\infty \myexp{-2 \pi^2 \sigma^2 y^2/T^2} dy \\
&= \frac{2}{\sqrt{\pi}} \int_{\sqrt{2}\pi \sigma N/T}^\infty \myexp{-z^2} dz \\
&= \erfc(\sqrt{2}\pi \sigma N/T) \\
&\le \myexp{-2 \pi^2 \sigma^2 N^2 / T^2}. 
\end{align}
By picking some $N=\mathcal{O}(\sigma^{-1} T \sqrt{\mylog{1/\epsilon}})$, we can guarantee that
$\myexp{-2 \pi^2 \sigma^2 N^2 / T^2} \le \epsilon/6$ and hence Eq. \eqref{eq:trunc_err2} holds for all $x \in \R$.

\end{proof}

\section{Certification of ground state energy estimates}\label{app:certification}

Our algorithm has the advantageous property of requiring only circuit depths that are inversely proportional to the spectral gap of the Hamiltonian. However, in many practical applications, the spectral gap may not be known in advance, and it may be necessary to verify the accuracy of the algorithm’s output.

Moreover, as we will explain shortly, depending on the structure of the initial state, the algorithm may still produce a correct estimate even with evolution times shorter than the inverse of the spectral gap. If we suspect that we have such a structured initial state, it is useful to have a method for verifying the correctness of the output.

We now outline a test that, given \( \sigma > 0 \), an estimate \( \hat{E}_0 \) of the ground state energy \( E_0 \) accurate within \( \sigma \), a precision tolerance \( \epsilon \), a failure probability \( \delta > 0 \), and a lower bound \( \eta \) on the initial overlap \( p_0 \), generates a new estimate $\hat{E}'_0$ of $E_0$, and certifies whether \( |\hat{E}'_0 - E_0| \leq \epsilon \), with a failure probability of at most \( \delta \) for incorrect certification. The quantum circuits required to run this algorithm have a depth of \( \tilde{\mathcal{O}}(\sigma^{-1}) \). Note that although the algorithm may occasionally reject correct estimates, it will consistently reject incorrect estimates with high probability. Moreover, we will demonstrate that if \( \sigma \) is proportional to the spectral gap \( \Delta_{\rm true} \), the algorithm is guaranteed to accept with probability at least \( 1 - \delta \), although it may also accept for larger values of \( \sigma \) (or shorter circuit depths) for more structured states.

Additionally, note that we do \emph{not} certify that the spectral gap of the Hamiltonian is at least some value \( \Delta \); rather, we certify only that the ground state energy estimate is accurate within \( \epsilon \), under the promise that the initial state has an overlap of at least \( \eta > 0 \) with the ground space. Certifying the spectral gap would require further assumptions about the initial state and its overlap with the first excited state. To clarify this point, consider the case where the initial state is \( \rho = \ketbra{E_0}{E_0} \), i.e., the ground state. In this scenario, it is straightforward to see that our algorithm yields the correct estimate with controlled evolution times of order \( \mathcal{O}(1) \). Thus, given this initial state and sufficient samples, our algorithm will always produce the correct estimate, regardless of circuit depth. Conversely, it is evident that the state \( \rho \) provides no information about the spectral gap of the Hamiltonian. In summary, our test will certify only that the energy estimate lies within an acceptable error margin.

Consider the mixture of Gaussians $p\ast n_\sigma$. If $\sigma$ is small enough to make sure that the contribution of the other energies in $n_\sigma$ at $E_0$ is sufficiently suppressed, then we have a peak close to the ground state energy. Otherwise, we have pollution from other energies, and they will make the peak around $E_0$ wider than if not. Thus, our test will consist on first making sure that we have a peak close to $\hat{E}_0$ and then measuring the width of the peak.

The following result about mixtures of Gaussians will be useful. It shows that if the mean of a Gaussian mixture \( p \ast n_\sigma \) deviates significantly from the minimum mean of its components, the collective variance will exceed the individual variances \( \sigma^2 \) of those components.

\begin{lemma}\label{lem:variance_exceeds}
Let $X=\sum_ip_iN_i$, where $N_i\sim\mathcal{N}(E_i,\sigma^2)$ are normal random variables with the same variance $\sigma^2$ but different means $E_i$. Furthermore, assume $E_0<E_1<\cdots<E_m$ and that $p_0\geq \eta$.
If $|\mathbb{E}(X)-E_0|\geq c\epsilon$ for some constants $c>1,\epsilon>0$, then:
\begin{align}\label{equ:variance_violated}
\mathbb{E}(|X-\mathbb{E}(X)|^2)\geq \sigma^2+\eta c^2\epsilon^2
\end{align}
\end{lemma}
\begin{proof}
Recall the standard identity $\mathbb{E}(|X-\mathbb{E}(X)|^2)=\mathbb{E}(X^2)-\mathbb{E}(X)^2$. Let us compute these expectation values. Clearly, $\mathbb{E}(X)=\sum_ip_iE_i$. Furthermore,
\begin{align}
\mathbb{E}(X^2)=\sum_i p_i\mathbb{E}(N_i^2)=\sum_ip_i(\sigma^2+E_i^2).
\end{align}
Thus,
\begin{align*}
    \mathbb{E}(X^2)-\mathbb{E}(X)^2=\sigma^2+\sum_ip_iE_i^2-(\sum_ip_iE_i)^2.
\end{align*}
Now note that $\sum_ip_iE_i^2-(\sum_ip_iE_i)^2$ is nothing but the variance of a discrete random variable $Y$ that takes value $E_i$ with probability $p_i$ and $\mathbb{E}(Y)=\mathbb{E}(X)$. Then we have
\begin{align}
    \mathbb{E}(|Y-\mathbb{E}(Y)|^2)\geq \eta c^2\epsilon^2.
\end{align}
To see this, note that by our assumption that $|\mathbb{E}(X)-E_0|\geq c\epsilon$, with probability $p_0$ the difference $|Y-\mathbb{E}(Y)|^2$ will be at least $c^2\epsilon^2$, which yields the claim.
\end{proof}

Our test will be based on ``zooming into'' the peak around the given  estimate and testing whether the variance is larger than expected. That is, we will condition the distribution around $\hat{E}_0$ and estimate its variance.

But it is not given that conditioning around $\hat{E}_0$ will also yield a distribution that is a mixture of Gaussians. Thus, we will show that this is approximately the case if we have a well-defined peak in the distribution.

Before we prove our theorem showing how conditioning on an interval yields another mixture of Gaussians, let us again give some intuition. Given some peak $\hat{E}_0$ that we identified, we will truncate to some interval $[\hat{E}_0- L, \hat{E}_0+L]$, with 
\begin{align}\label{equ:size_L}
L=\sqrt{2c\myln{\sigma^{-1}\tau^{-1}}}\sigma
\end{align}
for some constants $c>1$, $0<\tau<1$ to be specified later. The mass of energies that are located to the right of $E_0+3L$ is negligible on the interval $[\hat{E}_0- L, \hat{E}_0+L]$, as it is at most $\myO{L\tau^{2c}}$. So we can safely discard these random variables when truncating to the interval. Furthermore, we will also approximate well all Gaussians whose means are located in $[\hat{E}_0- L/2, \hat{E}_0+L/2]$, as their mass outside this interval is $\myO{\tau^{c/4}}$. 

Thus, the only energies that ``pose a threat" of pollution are those that have substantial weight on the edge of our interval $[\hat{E}_0-L,\hat{E}_0+L]$. This is where we use the assumption that we found a peak: we will impose that the mass going around the edges of the intervals is small. Note that only the right edge of the interval has to be considered, as we assume in our algorithm that we know the ground state energy up to $\sigma$ and, thus, there are no energies to the left of the interval.
More formally, we will impose:
\begin{align}\label{equ:integral_small}
    \int_{\hat{E}_0+\tfrac{L}{2}}^{\hat{E}_0+2L}(p\ast n_\sigma)(x)dx\leq c'\epsilon^2\eta.
\end{align}
for some constant $c'$ we will choose later. We will call distributions that satisfy this property \emph{peaked around $\hat{E}_0$.}

\begin{definition}[Distribution peaked around value]
Given a point distribution $p$ on $\mathbb{R}$, $\epsilon>0$, $\eta>0$ and $n_\sigma$ the probability density function of a Gaussian random variable with standard deviation $\sigma$, we say that $p\ast n_\sigma$ is peaked around a point $\hat{E}_0$ up to $\epsilon$ with parameters $c,c',\tau>0$ if for $L=\sigma\sqrt{2c\myln{\sigma^{-1}\tau^{-1}}}$ we have that:
\begin{align}
    \int_{\hat{E}_0+\tfrac{L}{2}}^{\hat{E}_0+2L}(p\ast n_\sigma)(x)dx\leq c'\epsilon^2\eta.
\end{align}
\end{definition}

To further simplify notation we will define $n_{\sigma,i}$ to be the density of the random variable $\mathcal{N}(E_i,\sigma^2)$.
We will start by proving that whenever Eq.~\eqref{equ:integral_small} holds, the mass on $[\hat{E}_0- L, \hat{E}_0+L]$
is mostly given by energies in the interval 
$[\hat{E}_0- L/2, \hat{E}_0+L/2]$.
\begin{lemma}\label{lem:contribution_small}
Let $X=\sum_ip_iN_i$, where $N_i\sim\mathcal{N}(E_i,\sigma^2)$ are normal random variables with the same variance $\sigma^2$ but different means $E_i$ and $L$ be defined as in Eq.~\eqref{equ:size_L}.
For some $\hat{E}_0$ that is within $\sigma$ of $E_0$, group the means of the $N_i$ into three subsets:
\begin{align}
    G&=\{i:|E_i-\hat{E}_0| \le \tfrac{L}{2}\},\\
    B&=\{i:\tfrac{L}{2}<|E_i-\hat{E}_0|\leq 2L\},\\
    F&=\{i:2L<|E_i-\hat{E}_0|\}.
\end{align}
If $p\ast n_\sigma$ is peaked around $\hat{E}_0$ with $c=4$ and $\tau=\frac{c'\epsilon^2\eta}{4}$ then:
\begin{align}\label{equ:good_approximation}
\left|\int_{\hat{E}_0-L}^{\hat{E}_0+L}(p\ast n_\sigma)(x)dx-\sum_{i\in G}p_i\right|\leq 4c'\epsilon^2\eta.
\end{align}

\end{lemma}
\begin{proof}
Clearly:
\begin{align}
    \int_{\hat{E}_0-L}^{\hat{E}_0+L}(p\ast n_\sigma)(x)dx=\sum_{i\in G}p_i\int_{\hat{E}_0-L}^{\hat{E}_0+L}n_{\sigma,i}(x)dx+\sum_{i\in B}p_i\int_{\hat{E}_0-L}^{\hat{E}_0+L}n_{\sigma,i}(x)dx+\sum_{i\in F}p_i\int_{\hat{E}_0-L}^{\hat{E}_0+L}n_{\sigma,i}(x)dx.
\end{align}
We will estimate each term in the sum separately. Let us start with energies in $F$. As, by definition, all energies in that set are at least $2L$ away from $\hat{E_0}$, they are at least $L$ away from the upper end of the integration limits. Thus, by our choice of $L$ we have for $i\in F$:
\begin{align}
    \int_{\hat{E}_0-L}^{\hat{E}_0+L}n_{\sigma,i}(x)dx\leq 2L\frac{e^{c\mylog{\tau\sigma}}}{\sigma\sqrt{2\pi}}=\mathcal{O}(\tau^{c}\mylog{\tau^{-1}})
\end{align}
and so 
\begin{align}\label{equ:contribution_F}
    \sum_{i\in F}p_i\int_{\hat{E}_0-L}^{\hat{E}_0+L}n_{\sigma,i}(x)dx=\mathcal{O}(\tau^{c}\mylog{\tau^{-1}}).
\end{align}
Let us now estimate the contribution of terms in $G$. We have for $i\in G$:
\begin{align}
\int_{\hat{E}_0-L}^{\hat{E}_0+L}n_{\sigma,i}(x)dx=1- \int_{-\infty}^{\hat{E}_0-L}n_{\sigma,i}(x)dx-\int_{\hat{E}_0+L}^{+\infty}n_{\sigma,i}(x)dx.
\end{align}
By the property that all energies in $G$ are $\tfrac{L}{2}$ away from $\hat{E}_0$, we see that both integrals above are at most $\mathcal{O}(\tau^{\tfrac{c}{2}})$. Thus:
\begin{align}\label{equ:contribution_G}
    \left|\sum_{i\in G}p_i\int_{\hat{E}_0-L}^{\hat{E}_0+L}n_{\sigma,i}(x)dx-\sum_{i\in G}p_i\right|=\mathcal{O}(\tau^{\tfrac{c}{2}}).
    \end{align}
By our choice of $c$ and $\tau$, we have  $\tau^{\tfrac{c}{2}}=\left(\tfrac{c'\epsilon^2\eta}{4}\right)^2$. Thus, both  Eq.~\eqref{equ:contribution_F} and Eq.~\eqref{equ:contribution_G} are bounded by $c'\epsilon^2\eta$.

To finish the proof, we need to consider the contribution of the energies in $B$. Note that for $i\in B$:
\begin{align}
    \sum_{i\in B}p_i\int_{\hat{E}_0-L}^{\hat{E}_0+L}n_{\sigma,i}(x)dx\leq 2\sum_{i\in B}p_i\int_{\hat{E}_0+\tfrac{L}{2}}^{\hat{E}_0+2L}n_{\sigma,i}(x)dx\leq \sum_{i}2p_i\int_{\hat{E}_0+\tfrac{L}{2}}^{\hat{E}_0+2L}n_{\sigma,i}(x)dx\leq 2c'\epsilon^2\eta.
\end{align}
To see this, note that by definition of $B$, the mean of these Gaussians are greater than $\hat{E}_0+\frac{L}{2}$, and the integral on the RHS is over a larger interval that is closer to their mean. Thus, as the mass of the Gaussian is maximal around the mean, the inequality follows. By our assumption in Eq.~\eqref{equ:integral_small} we conclude that:
\begin{align}\label{equ:contribution_B}
   \sum_{i\in B}p_i \int_{\hat{E}_0-L}^{\hat{E}_0+L}n_{\sigma,i}(x)dx\leq 2c'\epsilon^2\eta.
\end{align}
Putting together the inequalities in Eqs.~\eqref{equ:contribution_F},~\eqref{equ:contribution_G} and~\eqref{equ:contribution_B}, Eq.~\eqref{equ:good_approximation} follows.

\end{proof}

We can show in a similar fashion that the following holds:
\begin{corollary}\label{cor:moments_similar}
Under the same conditions as Lemma~\ref{lem:contribution_small}, assume further that $\sigma$ is small enough to ensure that $|\hat{E}_0\pm L|\leq 1$.
If $p\ast n_\sigma$ is peaked around $\hat{E}_0$ with $c=4$ and $\tau=\frac{c'\epsilon^2\eta}{4}$ then:
\begin{align}\label{equ:good_approximation_expect}
\left|\int_{\hat{E}_0-L}^{\hat{E}_0+L}x (p\ast n_\sigma)(x)dx-\sum_{i\in G}p_iE_i\right|\leq 4c'\epsilon^2\eta.
\end{align}
and 
\begin{align}\label{equ:good_approx_variance}
    \left|\int_{\hat{E}_0-L}^{\hat{E}_0+L}x^2(p\ast n_\sigma)(x)dx-\sum_{i\in G}p_i(\sigma^2+E_i^2)\right|\leq 4c'\epsilon^2\eta.
\end{align}
\end{corollary}
\begin{proof}
We can apply the same estimates as we used to show Lemma~\ref{lem:contribution_small} to show that the contributions from energies in $B$ and $F$ are small. This is because we assume that $[\hat{E}_0-L,\hat{E}_0+L]\subset[-1,1]$ and for $x\in[-1,1]$:
\begin{align}
    |xn_{\sigma,i}|\leq n_{\sigma,i},\quad x^2n_{\sigma,i}\leq n_{\sigma,i}
\end{align}
Furthermore, it is also not difficult to show that for $i\in G$ the value of the first two moments are approximated if we integrate over $[\hat{E}_0- L, \hat{E}_0+L]$.
\end{proof}

This shows that under Eq.~\eqref{equ:integral_small} the moments of the random variable conditioned on the interval $[\hat{E}_0- L, \hat{E}_0+L]$ are indeed well-approximated by those of the energies in the interval $[\hat{E}_0- L/2, \hat{E}_0+L/2]$. Thus, if we verify that Eq.~\eqref{equ:integral_small} holds and estimate the moments conditioned on the interval, then we can approximate the quantity in Eq.~\eqref{equ:variance_violated} and determine whether the variance is larger than expected, assuming the energy estimate is correct up to $c\epsilon$.

\subsection{Certification algorithm and its correctness}
With these technical results at hand, let us give a description of our energy certification in Algorithm~\ref{alg:certification}, which we refer to as $\mathrm{GSEE}\_\mathrm{CERT}$.
\begin{algorithm}[ht]

    \caption{Certification of ground state energy estimate}
    \begin{algorithmic}[1]

    \Procedure{GSEE\_CERT}{$\epsilon$, $\hat{E}_0$, $\sigma$, $\eta$, $p$}
        \State $\tau\gets \frac{\epsilon^{2}\eta}{160}$
        \State Let $X\sim\sum_ip_iN(E_i,\sigma^2)$ with pdf $p\ast n_\sigma$
        \State $L=\sqrt{8\myln{\sigma^{-1}\tau^{-1}}}\sigma$
        \State  $\tau'\gets$ estimate of $\int_{\hat{E}_0+\tfrac{L}{2}}^{\hat{E}_0+2L}(p\ast n_\sigma)(x)dx$ up to $\tfrac{\epsilon^2\eta}{160}$.
        \If{$\tau'<\tfrac{\epsilon^2\eta}{80}$}
        \State $M\gets$ estimate $\mathbb{E}(X|X\in[\hat{E}_0-L,\hat{E}_0+L])$ up to $\tfrac{\epsilon^2\eta}{160}$
        \State $S\gets$ estimate $\mathbb{E}((X-M)^2|X\in[\hat{E}_0-L,\hat{E}_0+L])$ up to $\tfrac{\epsilon^2\eta}{160}$.
        \If{$|S-\sigma^2|\leq2\epsilon^2\eta$}
        \State \Return Accept estimate and return $M$.
        \Else 
        \State \Return Reject estimate.
        \EndIf
        \Else 
        \State \Return Reject estimate.\label{line:if}
        \EndIf
    \EndProcedure
    \end{algorithmic}
            \label{alg:certification}

    \end{algorithm}
We deliberately formulated it in a way independent of how we obtained the estimates for the relevant quantities. This is because different algorithms for GSEE have different access models to the function $p\ast n_\sigma$, be it samples, or to estimates of the density, as in Ref.~\cite{wang2023quantum}. Furthermore, we assume that we are certain about the correctness of the estimates used in the algorithm for simplicity. But it is straightforward to extend the analysis to the case where the estimates are only correct up to some given failure probability.
We then have:
\begin{theorem}
Suppose that $\hat{E}_0$ satisfies $|\hat{E}_0-E_0|\leq \sigma$, Algorithm~\ref{alg:certification} accepts, and all estimates within Algorithm~\ref{alg:certification} are accurate. Then we have
\begin{align}\label{equ:good_estimate}
|M-E_0|\leq 4\epsilon.
\end{align}
\end{theorem}
\begin{proof}
If the algorithm accepted, then we already know that the density $p\ast n_\sigma$ is peaked around $\hat{E}_0$. Otherwise we would have already rejected in the first at Line~\ref{line:if}.
Let us assume that 
\begin{align}\label{equ:bad_estimate}
    |M-E_0|> 4\epsilon.
\end{align}
We will now show that the test will reject. As before, let $X=\sum_ip_i N(E_i,\sigma^2)$ and $\tilde{X}$ be the same random variable conditioned on the interval $[\hat{E}_0-L,\hat{E}_0+L]$. Furthermore, also let as before $G$ be the set of energies $E_i\in[\hat{E}_0-L/2,\hat{E}_0+L/2]$.
By Lemma~\ref{lem:contribution_small} and Corollary~\ref{cor:moments_similar}, we have that the variance of $\tilde{X}$ approximates up to $\frac{\epsilon^2\eta}{160}$ that of the random variable $Y=\left(\sum_{i\in G}p_i N(E_i,\sigma^2)\right)/\left(\sum_{i\in G}p_i\right)$.

By Lemma~\ref{lem:variance_exceeds}, if Eq.~\eqref{equ:bad_estimate} holds, then the variance of $X$ conditioned on $[\hat{E}_0-L,\hat{E}_0+L]$ must exceed 
\begin{align}
\mathbb{E}(|X-\hat{E}_0|^2|X\in[\hat{E}_0-L,\hat{E}_0+L])\geq \sigma^2+16\eta\epsilon^2-\frac{\eta\epsilon^2}{160}>\sigma^2+2\eta\epsilon^2,
\end{align}
which means that the test will have rejected.
Thus, if we accept, then Eq.~\eqref{equ:good_estimate} must hold.

\end{proof}

We now demonstrate that the test will accept if $\sigma$ is smaller than the spectral gap:

\begin{lemma}
Let $\Delta_{\operatorname{true}}$ be the spectral gap of $H$, $p$ be a spectral measure with $p_0\geq \eta$, $\epsilon>0$, and $|\hat{E}_0-E_0|\leq \sigma$. Let $\sigma_0$ be the largest $\sigma>0$ such that Algorithm~\ref{alg:certification} accepts when run with arguments   $(\epsilon, \hat{E}_0, \sigma, \eta, p)$. Then:
\begin{align}\label{equ:maximal_variance}
 \sigma_0 \geq \frac{\Delta_{\operatorname{true}}}{10\sqrt{\myln{2\eta^{-1}\epsilon^{-1}}}}.
\end{align}
\end{lemma}
\begin{proof}
Let us first show that for our value of $\sigma_0$ in Eq.~\eqref{equ:maximal_variance}, the distribution will be peaked around $E_0$. By our choice of $\sigma_0$, it is easy to see  that the contribution of the other energies on $[E_0-2L,E_0+2L]$ will be of order $\mathcal{O}((\epsilon\eta)^{50})$, which immediately implies that the distribution is peaked.
Furthermore, the probability distribution on $[E_0-2L,E_0+2L]$ is well-approximated by a single Gaussian up to corrections of order $\mathcal{O}((\epsilon\eta)^{50})$. Thus, the variance conditioned on $[E_0-L,E_0+L]$ will deviate from $\sigma^2$ by less than $\epsilon^2\eta$. As our estimate $S$ of the variance deviates from the true one by at most $\frac{\epsilon^2\eta}{160}$, we have $|S-\sigma^2|\leq 2\epsilon^2\eta$, which means we will accept.
\end{proof}

We conclude that the test is sound: in the worst case, if we select \( \sigma \) with a magnitude similar to the spectral gap, acceptance is guaranteed. However, it may also be the case that \( \sigma_0 \gg \Delta_{\operatorname{true}} \), depending on the initial state. To illustrate this, consider once again the case where \( p_0 = 1 \). In this scenario, the test will accept regardless of the gap size. This occurs because \( n_\sigma \) will contain only a single Gaussian component, resulting in a constant variance of \( \sigma^2 \).

\subsection{Sample complexity of $\mathrm{GSEE}\_\mathrm{CERT}$ with rejection sampling}
Let us analyze the sample complexity of the test described above, given access to samples and a success probability of correct verification of $1-\delta$.

\begin{corollary}
Under the same conditions as Theorem \ref{thm:main_in_he}, assume further that \( |E_0 - \hat{E}_0| \leq \sigma \). Then we can perform Algorithm~\ref{alg:certification} with a probability of correct rejection or acceptance of \( 1 - \delta \), using
\begin{align}
\mytO{\epsilon^{-4}\eta^{-3}\mylog{\delta^{-1}}}
\end{align}
block-encodings of Gaussian functions of $H$, as provided in Corollary \ref{cor:qetu_circuit_shifted_gaussian},

\end{corollary}
\begin{proof}
There are two steps that need to be performed for running Algorithm~\ref{alg:certification}. First, we need to determine if the distribution is peaked around $\hat{E}_0$. If it is peaked, we need to determine the variance up to an error $\myO{\eta\epsilon^2}$ on the interval $[\hat{E}_0-L,\hat{E}_0+L]$.
Let us start by showing how to determine if the distribution is peaked. For that we will perform rejection sampling with the Gaussian conditioned on $[\hat{E}_0+\tfrac{L}{2},\hat{E}_0+2L]$. As explained in Section~\ref{subsec:rejection_sampling} in Eq.~\eqref{equ:constant_rejection}, the probability of accepting a sample is given by 
\begin{align}
\frac{\sqrt{2\pi} \sigma \int_{\hat{E}_0+L/2}^{\hat{E}_0+2L}(p \ast n_\sigma)(x) dx}{3L/2}
=  \frac{\sqrt{2\pi}  \int_{\hat{E}_0+L/2}^{\hat{E}_0+2L}(p \ast n_\sigma)(x) dx}{3\sqrt{2 \myln{\sigma^{-1}\tau^{-1}}}}.
\end{align}
Thus, by running $\mytO{\epsilon^{-2}\eta^{-1}\mylog{\delta^{-1}}}$ rounds of the rejection sampling, we can decide if the distribution is peaked around $\hat{E}_0$ or not with probability of failure at most ${\delta}/{2}$ by checking the fraction of accepted samples.

If it is peaked, then we can perform rejection sampling again, but now on the interval $[\hat{E}_0-L,\hat{E}_0+L]$. With access to $\mathcal{O}(\eta^{-2}\epsilon^{-4}\mylog{\delta^{-1}})$ samples from the conditioned distribution, we can estimate the variance of the conditioned random variable to a precision of $\mathcal{O}(\epsilon^2\eta)$, with a failure probability of at most ${\delta}/{2}$. 

The probability of accepting a sample from the interval is given by
\begin{align}
\frac{\sqrt{2\pi} \sigma \int_{\hat{E}_0-L}^{\hat{E}_0+L}(p \ast n_\sigma)(x) dx}{2L}
=  \frac{\sqrt{2\pi} \int_{\hat{E}_0-L}^{\hat{E}_0+L}(p \ast n_\sigma)(x) dx}{2\sqrt{8 \myln{\sigma^{-1}\tau^{-1}}}}.    
\end{align}
By our choice of $L$ and the assumption that $|E_0-\hat{E}_0|\leq\sigma$, we have $E_0\in[\hat{E}_0-L,\hat{E}_0+L]$. It follows 
that 
\begin{align}
\int_{\hat{E}_0-L}^{\hat{E}_0+L}(p\ast n_{\sigma})(x)dx=\Omega(\eta).    
\end{align}
Thus, we need to run the rejection sampling an expected \( \tilde{\mathcal{O}}(\eta^{-1}) \) times to accept a sample. This results in a total complexity  of \( \tilde{\mathcal{O}}(\eta^{-3} \epsilon^{-4} \log{\delta^{-1}}) \) for this step, which dominates the overall sample complexity of the certification algorithm. Additionally, the probability of obtaining incorrect estimates is at most \( \delta \). Since the algorithm always produces a  correct output if given accurate estimates, this completes the proof of the claim.
\end{proof}
Note that the assumption that we are given an estimate that satisfies $|E_0-\hat{E}_0|\leq\sigma$ does not increase the overall depth of circuits required for the certification method. This is because we can always obtain such an estimate with circuits of depth $\mytO{\sigma^{-1}}$ using e.g. Algorithm~\ref{alg:basic_gsee}.

\end{document}